\documentclass[letterpaper, 10 pt, journal]{IEEEtran}
\IEEEoverridecommandlockouts

\usepackage{cite}
\usepackage{amsmath,amssymb,amsfonts,amsthm}
\allowdisplaybreaks[4]
\usepackage{algorithm}
\usepackage{algpseudocode}
\algrenewcommand\algorithmicrequire{\textbf{Input:}}
\algrenewcommand\algorithmicensure{\textbf{Output:}}
\usepackage{graphicx}
\usepackage{booktabs}
\usepackage{float}
\usepackage{multirow}

\newtheorem{proposition}{Proposition}
\usepackage{tikz}
\usetikzlibrary{backgrounds,arrows.meta,positioning,calc,shapes.geometric,fit}
\usepackage{hyperref}

\begin{document}

    \bstctlcite{BSTcontrol}

    \title{
        \LARGE \bfseries RFSR-MI-MPC: Ranking-Based Feasible-Set Restriction for Real-Time Control of Reconfigurable Battery Packs
    }
    
    \author{
        Albert~Škegro, \IEEEmembership{Student Member,~IEEE,} Quan~Ouyang, \IEEEmembership{Member,~IEEE,} Torsten~Wik, \IEEEmembership{Member,~IEEE,}
        Changfu~Zou, \IEEEmembership{Senior Member,~IEEE}
 
        \thanks{This work was supported by the Swedish Research Council (Grant No.~2023--04314) and the European Union's Horizon Europe programme through the Marie Sk\l{}odowska--Curie Actions (Grant No.~101131278). The computations were enabled by resources provided by the National Academic Infrastructure for Supercomputing in Sweden (NAISS) at the Chalmers University of Technology, partially funded by the Swedish Research Council (Grant No.~2022--06725).}%
        
        \thanks{Albert~Škegro, Torsten~Wik, and Changfu~Zou are with the Department of Electrical Engineering, Chalmers University of Technology, Gothenburg, Sweden. (e-mails: skegro@chalmers.se, tw@chalmers.se, changfu.zou@chalmers.se).}
        \thanks{Quan~Ouyang is with the College of Automation Engineering, Nanjing University of Aeronautics and Astronautics, Nanjing, China (e-mail: ouyangquan@nuaa.edu.cn)}
    }

    \markboth{Submitted to IEEE Transactions on Control Systems Technology}%
    {Škegro \MakeLowercase{\textit{et al.}}: RFSR-MI-MPC: Ranking-Based Feasible-Set Restriction}
    \maketitle

    \begin{abstract}
        A ranked-prefix feasible-set restriction is proposed for real-time mixed-integer model predictive control of series-connected reconfigurable battery packs with per-cell bypass switching. Cells are ordered at each sampling instant by a one-step objective-informed suitability score, and the engagement variables are constrained to prefixes of that ordering. The restriction reduces the number of admissible per-step engagement patterns from combinatorial to linear in the number of cells while leaving the model, objective, and physical constraints unchanged. For a 20-cell pack driven over a worldwide harmonized light vehicles test cycle, the proposed controller certifies every active step, solves every instance without branching, and reduces the 95th-percentile solver time by a factor of 6.8 relative to the full subset-selection controller. The cycle-mean cell-to-cell state-of-charge standard deviation decreases by a factor of 3.4, and no power curtailment occurs. Same-state reference solves show a mean restriction gap below 0.11\% on jointly certified steps, and the computational gains persist across the tested initial conditions, cell heterogeneity levels, pack sizes, and prediction horizons without retuning.
    \end{abstract}


    \section{Introduction}
        
        Online mixed-integer model predictive control (MI-MPC) is a natural framework for constrained control problems in which continuous inputs must be chosen jointly with discrete operating decisions. Its appeal is clear: hybrid dynamics, logic conditions, physical limits, and performance objectives can be enforced in a single receding-horizon optimization problem~\cite{rawlings2020model,bemporad1999control,richards2005mixed}. Its practical limitation is equally clear. At each sampling instant, the controller must solve a mixed-integer program (MIP) within a fixed time budget. Under this limit, branch-and-bound may return a feasible incumbent without certifying optimality to the prescribed tolerance~\cite{marcucci2020warm,hespanhol2019structure}. The resulting instance-dependent time required for certification is a central obstacle to reliable real-time MI-MPC.
        
        Series-connected reconfigurable battery packs (RBPs) with per-cell bypass switching make this obstacle especially pronounced. By engaging or bypassing individual cells, such packs can improve balancing, fault tolerance, and utilization relative to fixed interconnections~\cite{han2020next,muhammad2019reconfigurable,ci2016reconfigurable}. The same flexibility, however, creates a large binary decision set. If an $N$-cell pack engages $m$ cells, then $\binom{N}{m}$ same-cardinality subsets are admissible. Under weak cell heterogeneity, many same-cardinality subsets have similar pack-level effects because cells with comparable SOC, temperature, resistance, and aging condition produce similar terminal voltages and predicted state responses. The computational burden therefore arises not only from the size of the feasible set, but also from the need to distinguish among many combinatorially distinct patterns that often provide little additional control authority. This near-equivalence is typical rather than guaranteed for every admissible subset. 

        Battery-pack management has traditionally relied on cell balancing to counter cell-to-cell variation. Passive balancing dissipates excess charge from higher-charge cells through resistive elements, and is simple and low-cost but energy-inefficient~\cite{hoque2017battery}. Active balancing instead redistributes charge among cells using switched-capacitor, inductor-based, or converter-based circuits, improving efficiency at the cost of additional hardware and control complexity~\cite{gallardo2014battery}. Reconfiguration, as introduced above, is a distinct and complementary paradigm to both: rather than redistributing charge within a fixed topology, it changes the pack's electrical topology itself. Prior work has shown that such reconfiguration can extend usable battery lifetime by over 20\%, particularly in high-voltage applications such as electric trucks and long-range passenger vehicles~\cite{vskegro2026system}.
        
        Existing work confirms the value of predictive optimization for reconfigurable batteries, but it also highlights the computational burden created by binary switching decisions. In modular or converter-integrated architectures, receding-horizon convex optimization has been used successfully when the dominant online decision is continuous power allocation among modules~\cite{farakhor2022novel,farakhor2023scalable}. For bypass-switched packs, however, predictive control faces a substantially harder subset-selection problem. Nonlinear MPC has been demonstrated on small reconfigurable packs with electrothermal constraints, using a 3~s prediction horizon and a solver time budget capped at the control sampling interval so that a feasible input is always available, at the cost of forgoing a global-optimality guarantee~\cite{mondoha2021nonlinear}. For larger commercial reconfigurable systems, mixed-integer MPC improves current-ripple suppression by 25\% relative to rule-based control (RBC), while state-of-charge (SOC) balancing, state-of-health (SOH) balancing, and reference tracking are comparable between the two; the RBC is at least 100 times faster~\cite{pinter2025comparative}.
        
        A different response in the literature is to reduce online complexity by adopting an alternative control paradigm that avoids solving the original mixed-integer predictive problem. Control allocation combined with control barrier functions can regulate output voltage, promote SOC balancing, and enforce electrothermal safety with low online cost~\cite{ebrahimi2026safe}. Reinforcement-learning and related learning-based controllers have also been proposed for real-time balancing, capacity utilization, and fault-tolerant reconfiguration, but typically do not provide per-step optimality certificates for the original MI-MPC formulation or guarantee hard constraint satisfaction, for example, cell over-voltage or over-temperature limits, since constraints are usually enforced only implicitly through the reward design~\cite{wei2026reliable,liu2026learning,irshayyid2026realtime}. These approaches are attractive when fast online evaluation is the dominant requirement, but they trade certified guarantees for computational simplicity.
        
        A separate line of work retains the original optimization formulation and instead improves the model or the solver. Unified modeling frameworks for reconfigurable battery systems represent switch-dependent interconnections in a control-oriented form and support mixed-integer optimal-control formulations~\cite{geng2025fundamental,geng2026unified}. Related search-space construction methods remove infeasible or topologically duplicate configurations, such as open circuits, short circuits, and invalid switch states~\cite{geng2025fundamental}. At the solver level, warm-start strategies and branch-and-bound reuse exploit the receding-horizon structure of MI-MPC, while structure-exploiting branch-and-bound methods leverage block-sparse optimal-control structure to reduce solution time~\cite{marcucci2020warm,hespanhol2019structure}. These advances are complementary, but they do not address a residual source of difficulty that remains in a bypass-switched series pack after infeasible and duplicate topologies have been excluded: the large number of feasible same-cardinality engagement subsets that a solver must still distinguish.
        
        This paper targets that residual formulation-level burden. We propose a ranking-based feasible-set restriction for MI-MPC, referred to as RFSR-MI-MPC. At each sampling instant, cells are ordered by an objective-derived sensitivity score, and only prefix selections in that ordering are admitted, replacing arbitrary same-cardinality subset selection with a ranked-prefix engagement set. In this work, the system model, ranking score, and restricted engagement sets are developed specifically for reconfigurable battery packs. The underlying restriction principle may also be applicable to structured MI-MPC problems in which discrete configurations can be ranked and represented by nested candidate subsets. Importantly, the restriction is designed to leave the model, objective, and constraints unchanged. The resulting controller remains a mixed-integer convex MPC scheme and provides optimality certificates for the restricted online problem, without claiming global optimality for the excluded non-prefix subsets. Overall, the major contributions of this paper are twofold:
        \begin{itemize}
            \item RFSR-MI-MPC, a ranked-prefix feasible-set restriction for structured MI-MPC, together with an optimality guarantee for the restricted problem under an affine one-step surrogate
            \item An instantiation of RFSR-MI-MPC for bypass-switched reconfigurable battery packs, comprising an objective-derived cell-ranking score and a restricted mixed-integer convex MPC formulation, evaluated in closed loop and validated against same-state full subset-selection reference solves.
        \end{itemize}        
        
        The remainder of the paper is organized as follows. Section~\ref{sec:modelAndProblemFormulation} presents the system model and MI-MPC formulation. Section~\ref{sec:restriction} introduces the combinatorial engagement structure, objective-informed ranking, ranked-prefix restriction, and its properties. Section~\ref{sec:evalMethodology} specifies the evaluation methodology, including the simulation and solver settings, performance metrics, and experimental protocols used to obtain the results. Section~\ref{sec:results} presents the closed-loop, reference-solve, Monte Carlo, and robustness results. Section~\ref{sec:conclusion} concludes the paper.

    \section{System Modeling and MI-MPC Formulation}
    \label{sec:modelAndProblemFormulation}

        We consider an $N$-unit series-connected reconfigurable battery pack, as shown in Fig.~\ref{fig:pack_topology}. Each unit $i$ contains one battery cell, one series switch whose on/off state is denoted by $S_i$, and one complementary bypass switch whose on/off state is denoted by $S^{\prime}_i$. At  discrete-time step $k$ with sampling time \(\Delta t\), the binary switching states satisfy
        \begin{subequations}
            \label{eq:switching_definition}
            \begin{align}
                S_i(k),\,S_i^{\prime}(k) &\in \{0,1\},\\
                S_i(k)+S_i^{\prime}(k) &= 1 \label{eq:S-binary-def},
            \end{align}
        \end{subequations}
        where \(S_i=1\) means that the series switch is on and cell $i$ is engaged in the series string, while \(S_i^{\prime}=1\) means that the bypass switch is on and cell \(i\) is bypassed. Interlocked gate signals enforce the complementary switching condition in \eqref{eq:S-binary-def}, which is treated as exact at the sampling instant $k$. The binary engagement vector is defined as
        \begin{align}
            S:=[S_1,\ldots,S_N]^{\top}.
        \end{align}
        
        \begin{figure*}[!ht]
            \centering
            \input{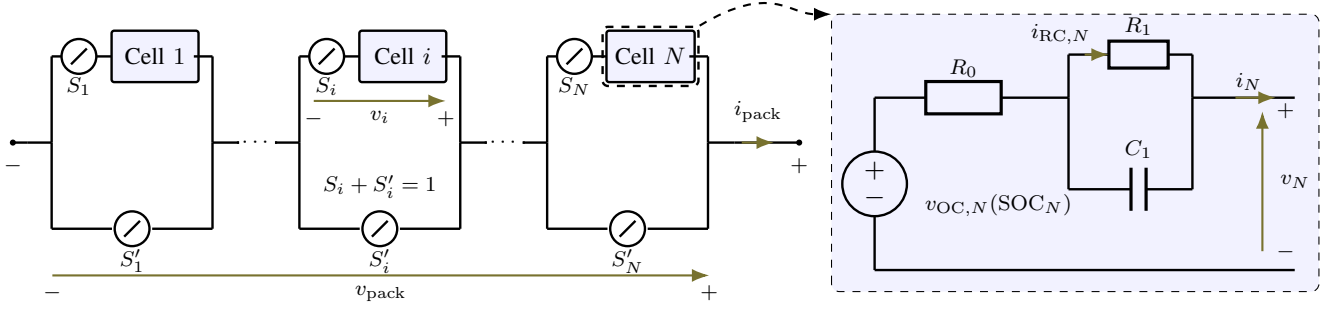}
            \caption{Series-connected reconfigurable battery pack with per-cell bypass switching and first-order resistor-capacitor (RC)-pair equivalent-circuit model used for each cell.}
            \label{fig:pack_topology}
        \end{figure*}

        \subsection{Electrothermal Battery Model}
        \label{subsec:electhermBattModel}

            The electrothermal dynamics of each cell are described by a first-order resistor-capacitor (RC) equivalent circuit coupled to a lumped thermal model~\cite{vskegro2023analysis,ouyang2025mathematical,allafi_lumped_2018}. The state variables of cell \(i\) are its state of charge \(\mathrm{SOC}_i\), RC-branch current \(i_{\mathrm{RC},i}\), and temperature \(T_i\). Their discrete-time dynamics are given by
            \begin{subequations}
                \label{eq:cellDynamicsExpanded}
                \begin{align}
                    \mathrm{SOC}_i(k+1)
                    =\:&
                    \mathrm{SOC}_i(k)
                    -
                    \frac{\Delta t}{3600Q_i} i_i(k),                                        \\
                    i_{\mathrm{RC},i}(k+1)
                    =\:&
                    a_1 i_{\mathrm{RC},i}(k)
                    +
                    b_1 i_i(k),                                                          \\
                    T_i(k+1)
                    =\:&
                    T_{\mathrm{amb}}
                    +
                    a_T\left[T_i(k)-T_{\mathrm{amb}}\right]
                    \nonumber\\
                    & +
                    b_T \dot Q_i^{\mathrm{gen}}(k),
                    \label{eq:cellDynamics_Thermal}
                \end{align}
            \end{subequations}
            where \(Q_i\) is the cell capacity, \(i_i\) is the cell current, \(T_{\mathrm{amb}}\) is the ambient temperature, and \(\dot Q_i^{\mathrm{gen}}\) is the heat-generation rate. The current is defined as positive during discharge and negative during charge. The model coefficients \(a_1\), \(b_1\), \(a_T\), and \(b_T\) are treated as constants and are determined by the underlying electrical and thermal parameters together with the sampling time $\Delta t$. Their values for the battery considered in this study are provided in the accompanying archival repository~\cite{skegroRankedPrefixArchive}.
            
            The pack current is denoted by \(i_{\mathrm{pack}}\). Because an engaged cell carries the pack current whereas a bypassed cell carries no current, the cell current satisfies
            \begin{equation}
                i_i(k)=S_i(k)\,i_{\mathrm{pack}}(k).
                \label{eq:cellCurrent}
            \end{equation}

            The heat-generation rate \(\dot Q_i^{\mathrm{gen}}\) is modeled as
            \begin{equation}
                \begin{aligned}
                    \dot Q_i^{\mathrm{gen}}(k)
                    &=
                    R_0 i_i^2(k)
                    +
                    R_1 i_{\mathrm{RC},i}^2(k)
                    +
                    r_s i_{\mathrm{pack}}^2(k),
                \end{aligned}
                \label{eq:qi_gen}
            \end{equation}
            where \(R_0\) and \(R_1\) are the ohmic and RC-branch resistances, and \(r_s\) is the switch on-state resistance. The RC-branch term represents internal polarization loss and is therefore retained when the cell is bypassed. Because the pack current passes through one conducting switch in every cell unit, the switch-loss term is independent of \(S_i\). The switch-generated heat in each unit is included in the corresponding lumped cell-unit thermal model.
            
            The terminal voltage of cell \(i\) is
            \begin{equation}
                v_i(k)
                =
                v_{\mathrm{OC},i}\bigl(\mathrm{SOC}_i(k)\bigr)
                -
                R_0 i_i(k)
                -
                R_1 i_{\mathrm{RC},i}(k),
                \label{eq:cellVoltage}
            \end{equation}
            where \(v_{\mathrm{OC},i}(\cdot)\) is represented by a piecewise-linear lookup map calibrated for the cell model~\cite{skegroRankedPrefixArchive}. The contribution of cell \(i\) to the series-string voltage is
            \begin{equation}
                v_{\mathrm{eng},i}(k)
                =
                S_i(k)\,v_i(k).
                \label{eq:engagedCellVoltage}
            \end{equation}
            The pack voltage \(v_{\mathrm{pack}}\) and power \(p_{\mathrm{pack}}\) are therefore
            \begin{subequations}
                \begin{align}
                    v_{\mathrm{pack}}(k)
                    &=
                    \sum_{i=1}^{N} v_{\mathrm{eng},i}(k)
                    -
                    N r_s i_{\mathrm{pack}}(k),
                    \label{eq:packVoltage}                                                   \\
                    p_{\mathrm{pack}}(k)
                    &=
                    i_{\mathrm{pack}}(k)\,v_{\mathrm{pack}}(k),
                    \label{eq:packPower}
                \end{align}
            \end{subequations}
            where the term \(N r_s i_{\mathrm{pack}}\) is the aggregate voltage drop across the \(N\) conducting switches.
            
            The pack-average SOC is
            \begin{equation}
                \bar{\mathrm{SOC}}(k)
                =
                \frac{1}{N}\sum_{i=1}^{N} \mathrm{SOC}_i(k).
            \end{equation}
            The SOH of cell \(i\) is defined by
            \begin{equation}
                \mathrm{SOH}_i
                =
                {Q_i}/{Q^{\mathrm{nom}}},
            \end{equation}
            where \(Q^{\mathrm{nom}}\) is the nominal cell capacity. In the MPC framework, \(Q_i\) and \(\mathrm{SOH}_i\) are treated as fixed and known parameters over an \(N_p\)-step prediction horizon, because degradation evolves on a much slower timescale than the electrothermal states.

        \subsection{MI-MPC Problem Formulation}
        \label{subsec:MI-MPCprobFor}
            We formulate an MI-MPC controller for the bypass-switched RBP. At each time step \(k\), an online optimization problem is solved over the prediction horizon \(\ell \in \{k+1,\ldots,k+N_p\}\) to minimize the objective function \(J(k)\). The objective penalizes three quantities: the use of degraded cells, SOC imbalance, and curtailed pack power. The decision variables include the binary engagement variables \(S_i\), the pack current \(i_{\mathrm{pack}}\), and the curtailment slack variable \(\lambda\). 
            For compactness, let
            \[
                \mathcal C := \{1,\ldots,N\},
                \qquad
                \mathcal H_k := \{k+1,\ldots,k+N_p\},
            \]
            denote the cell-index set and the prediction-horizon index set,
            respectively.  

            The direction of the requested power is denoted by
            \begin{equation}
                d_p(\ell)
                =
                \begin{cases}
                 1, & p_{\mathrm{req}}(\ell)\ge 0,\\
                -1, & p_{\mathrm{req}}(\ell)<0,
                \end{cases}
                \label{eq:powerDirection}
            \end{equation}
            where \(d_p=1\) during discharge and \(d_p=-1\) during charge.
            where \(d_p=1\) during discharge and \(d_p=-1\) during charge, with the discharge convention applied at \(p_{\mathrm{req}}(\ell)=0\). The curtailment slack variable $\lambda$ represents the unmet power expressed as an equivalent current:
            \begin{equation}
                \label{eq:lambdaDef}
                \lambda(\ell)
                =
                \begin{cases}
                    \dfrac{
                    p_{\mathrm{req}}(\ell)-p_{\mathrm{pack}}(\ell)
                    }{
                    d_p(\ell)v_{\mathrm{pack}}(\ell)
                    },
                    & p_{\mathrm{req}}(\ell)\neq 0, \\[2ex]
                    0,
                    & p_{\mathrm{req}}(\ell)=0 .
                \end{cases}
            \end{equation}

            The curtailment penalty is normalized by the operating-mode-dependent current scale \(i_{\mathrm{scale}}\), defined as
            \begin{equation}
                \label{eq:iscale}
                i_{\mathrm{scale}}(\ell)
                =
                \begin{cases}
                    i_{\mathrm{pack}}^{\max},
                    & p_{\mathrm{req}}(\ell)\ge 0,\\[0.5ex]
                    \lvert i_{\mathrm{pack}}^{\min}\rvert,
                    & p_{\mathrm{req}}(\ell)<0.
                \end{cases}
            \end{equation}
            where \(i_{\mathrm{pack}}^{\max}\) and \(i_{\mathrm{pack}}^{\min}\) are the discharge and charge current limits, respectively. The admissible curtailment is bounded by
            \begin{equation}
            \label{eq:lambdaBound}
                0 \le \lambda(\ell) \le \lambda^{\max}(\ell),
            \end{equation}
            with
            \begin{equation}
            \label{eq:lambdaMax}
                \lambda^{\max}(\ell)
                =
                \delta_P
                {|p_{\mathrm{req}}(\ell)|}/{v_{\mathrm{pack}}^{\min}},
            \end{equation}
            where \(\delta_P\in(0,1)\) specifies the fractional power-curtailment limit at \(v_{\mathrm{pack}}^{\min}\). When \(v_{\mathrm{pack}}(\ell)>v_{\mathrm{pack}}^{\min}\), the corresponding fractional power-curtailment bound increases proportionally to \(v_{\mathrm{pack}}(\ell)/v_{\mathrm{pack}}^{\min}\). Since \(v_{\mathrm{pack}}(\ell)>0\), the lower bound \(\lambda(\ell)\ge 0\) prevents power over-delivery during discharge and power over-absorption during charge.

            The objective function $J$ is defined as
            \begin{align}
                \label{eq:ContrOBJ}
                J(k)
                &=
                \frac{1}{N_p}
                \sum_{\ell=k+1}^{k+N_p}
                \Bigl[
                    w_{\mathrm{SOH}} J_{\mathrm{SOH}}(\ell)
                    +
                    w_{\mathrm{SOC}} J_{\mathrm{SOC}}(\ell)
                    \nonumber \\
                &\qquad\qquad\qquad
                    +
                    w_{\lambda} J_{\lambda}(\ell)
                \Bigr],
            \end{align}
            where \(w_{\mathrm{SOH}}, w_{\mathrm{SOC}}, w_{\lambda}\ge0\) are the corresponding weighting coefficients, and
            \begin{subequations}
                \label{eq:stageCosts}
                \begin{align}
                    J_{\mathrm{SOH}}(\ell)
                    &=
                    \frac{1}{N}
                    \sum_{i=1}^{N}
                    \frac{\mathrm{SOH}^{\mathrm{EOL}}}{\mathrm{SOH}_i}
                    S_i(\ell),
                    \label{eq:JSOH}
                    \\
                    J_{\mathrm{SOC}}(\ell) &= \frac{1}{N} \sum_{i=1}^{N} \left( \frac{\mathrm{SOC}_i(\ell+1)-\bar{\mathrm{SOC}}(\ell+1)} {\Delta \mathrm{SOC}^{\max}(k)} \right)^2, \label{eq:JSOC}
                    \\
                    J_{\lambda}(\ell)
                    &=
                    \frac{\lambda(\ell)}
                    {i_{\mathrm{scale}}(\ell)} .
                    \label{eq:Jlambda}
                \end{align}
            \end{subequations}
            Here, $\mathrm{SOH}^{\mathrm{EOL}}$ denotes the end-of-life SOH threshold, and $\Delta \mathrm{SOC}^{\max}$ is the terminal SOC-spread bound. $J_{\mathrm{SOH}}$ discourages the engagement of more degraded cells by assigning a larger cost to cells with lower SOH. $J_{\mathrm{SOC}}$ penalizes the cell-to-cell imbalance at the successor state resulting from the control action at stage \(\ell\). The curtailment term penalizes the unmet power request expressed in current-equivalent form. All the three stage-cost terms are dimensionless.

            To ensure that the pack current is consistent with the direction of the requested power, we define the admissible pack-current set \(\mathcal{I}\) as
            \begin{equation}
                \label{eq:currentSet}
                \mathcal{I}(\ell)
                =
                \begin{cases}
                    [0,\, i_{\mathrm{pack}}^{\max}],
                    & p_{\mathrm{req}}(\ell)>0, \\[0.5ex]
                    \{0\},
                    & p_{\mathrm{req}}(\ell)=0, \\[0.5ex]
                    [i_{\mathrm{pack}}^{\min},\,0],
                    & p_{\mathrm{req}}(\ell)<0 .
                \end{cases}
            \end{equation}

            With the control input defined as $u=[ S_1,  \ldots, \allowbreak S_N,  i_{\mathrm{pack}}, \allowbreak \lambda]^{\top}\in\{0,1\}^{N}\times\mathbb{R}^{2}$, the state vector as $x=[\mathrm{SOC}_1,  i_{\mathrm{RC},1}, T_1, \allowbreak  \ldots, \allowbreak  \mathrm{SOC}_N, i_{\mathrm{RC},N}, T_N]^{\top}\in\mathbb{R}^{3N}$, and the output vector as $y = [v_1, \allowbreak \ldots, v_N, v_\textrm{pack}]^{\top}\in\mathbb{R}^{N+1}$, where $f(\cdot,\cdot)$ denotes the vectorized state-transition map obtained by stacking the per-cell dynamics~\eqref{eq:cellDynamicsExpanded} over $i\in\mathcal C$, the full MI-MPC problem over all feasible engagement subsets can be formulated as            
            \begin{subequations}
                \label{eq:unrestrictedMIMPC}
                \begin{align}
                    J_{\mathrm{full}}^\star(k)
                    &= \min_{\mathcal{U}_k} J(k)
                    \label{eq:unrestrictedMIMPC_obj} \\
                    \mathrm{s.t.}\qquad \nonumber &
                    \\
                    x(\ell) &= f\left( x(\ell-1), u(\ell-1)\right), 
                    \label{eq:vectorizedStateEquation}\\
                    \sum_{i=1}^{N} S_i(\ell)
                    &\ge m_{\min}(\ell),
                    \label{eq:commonCardinalityFloor} \\
                    i_{\mathrm{pack}}(\ell)
                    &\in \mathcal{I}(\ell),
                    \label{eq:packCurrentConstraints} \\
                    \mathrm{SOC}_i(\ell)
                    &\in [\textrm{SOC}_{\mathrm{cell}}^{\min},\mathrm{SOC}_{\mathrm{cell}}^{\max}],
                    \label{eq:socConstraints} \\
                    T_i(\ell)
                    &\in [T_{\mathrm{cell}}^{\min},T_{\mathrm{cell}}^{\max}],
                    \label{eq:thermalConstraints} \\
                    \begin{split}
                    \max_{i\in\mathcal{C}} \mathrm{SOC}_i(k\!+\!N_p\!+\!1)
                    &-\min_{i\in\mathcal{C}} \mathrm{SOC}_i(k\!+\!N_p\!+\!1) \\
                    &\le \Delta \mathrm{SOC}^{\max}(k),
                    \end{split}
                    \label{eq:termSOCSpreadConstraint}\\
                    v_i(\ell)
                    &\in [v_{\mathrm{cell}}^{\min},v_{\mathrm{cell}}^{\max}],
                    \label{eq:cellVoltageConstraints} \\
                    v_{\mathrm{pack}}(\ell)
                    &\in [v_{\mathrm{pack}}^{\min},v_{\mathrm{pack}}^{\max}],
                    \label{eq:packVoltageConstraints} 
                \end{align}
            \end{subequations}
            where the cell index \(i\in\{1,\ldots, N\}\) and the time index within the prediction horizon \(\ell\in\{k+1,\ldots,k+N_p\}\). In \eqref{eq:commonCardinalityFloor}, the left-hand side \(m(\ell):=\sum_{i=1}^{N}S_i(\ell)\) is the number of cells engaged at stage \(\ell\), and \(m_{\min}(\ell)\) denotes the minimum number of cells that must be engaged at that stage. \eqref{eq:commonCardinalityFloor}--\eqref{eq:packCurrentConstraints} constrain the control inputs; \eqref{eq:socConstraints}-\eqref{eq:termSOCSpreadConstraint} impose the state constraints; and \eqref{eq:cellVoltageConstraints}--\eqref{eq:packVoltageConstraints} impose  the output constraints. In \eqref{eq:termSOCSpreadConstraint}, the terminal SOC-spread bound is updated before solving the MPC problem. It is tightened as the weakest cell approaches the lower SOC limit and relaxed when the pack is sufficiently far from that limit. This adaptive terminal bound prevents excessive cell-to-cell SOC imbalance near the lower SOC limit, where a weak cell could otherwise become power-limiting. The explicit forms of all these state dynamic equation and constraints can be derived from \eqref{eq:S-binary-def}--\eqref{eq:packPower} and \eqref{eq:lambdaDef}--\eqref{eq:lambdaMax}. Under the multiple-shooting schedule used to solve the problem, the optimization-variable set \(\mathcal{U}_k\) contains the input and state sequences over the prediction horizon, namely \(\mathcal{U}_k=[u_{k+1},\allowbreak \ldots,u_{k+N_p}, x_{k+1},\ldots, x_{k+N_p}]^{\top}\).

            The remainder of this subsection specifies how $\Delta\mathrm{SOC}^{\max}(k)$, the right-hand side of~\eqref{eq:termSOCSpreadConstraint}, is computed at each sampling instant so that it tightens as the weakest cell approaches the lower SOC limit and relaxes otherwise. The unsmoothed target bound is
            \begin{equation}
                \label{eq:dSOCmaxTarget}
                \Delta \mathrm{SOC}_{\mathrm{tar}}^{\max}(k)
                =
                \Delta_{\mathrm{low}}
                +
                \bigl(\Delta_{\mathrm{high}}-\Delta_{\mathrm{low}}\bigr)\tau(k),
            \end{equation}            
            with
            \begin{equation}
                \label{eq:tauAdaptive}
                \tau(k)
                =
                \min\left\{
                \max\left\{
                \frac{\min_i \mathrm{SOC}_i(k)-z_1}{z_2-z_1},
                0
                \right\},
                1
                \right\}.
            \end{equation}
            Here, $z_1$ and $z_2$ are SOC thresholds on the weakest cell that delimit a transition band: below $z_1$ the terminal spread bound is tightened to its minimum value $\Delta_{\mathrm{low}}$, above $z_2$ it is relaxed to its maximum value $\Delta_{\mathrm{high}}$, and it varies linearly with the weakest cell's SOC in between. By construction, $\Delta_{\mathrm{low}}\le\Delta_{\mathrm{high}}$ and $z_1<z_2$.
            
            To prevent abrupt tightening, the target is subject to the one-sided slew-rate limit \(\sigma_{\Delta}\):
            \begin{equation}
                \label{eq:dSOCmaxSlew}
                \begin{aligned}
                    \Delta \mathrm{SOC}^{\max}(k)
                    =\max\bigl\{&
                        \Delta \mathrm{SOC}_{\mathrm{tar}}^{\max}(k),\\
                        &\Delta \mathrm{SOC}^{\max}(k-1)
                        -\sigma_{\Delta}\Delta t
                    \bigr\}.
                \end{aligned}
            \end{equation}
            At controller initialization, the previous bound is set to \(\Delta \mathrm{SOC}^{\max}(k-1)=\Delta_{\mathrm{high}}\). The resulting value enters \eqref{eq:unrestrictedMIMPC} as a known parameter.            

            Problem \eqref{eq:unrestrictedMIMPC} defines the full subset-selection formulation subject to the common cardinality floor. Among constraints~\eqref{eq:commonCardinalityFloor}--\eqref{eq:packVoltageConstraints}, only the cardinality floor~\eqref{eq:commonCardinalityFloor} constrains which binary engagement patterns $S(\ell)$ are admissible; the remaining constraints act on the continuous state, input, and output trajectory. The full subset-selection formulation therefore permits every engagement pattern satisfying~\eqref{eq:commonCardinalityFloor}.

        \subsection{Mixed-Integer Convex Reformulation}
        \label{subsec:micp}
            The MI-MPC problem in Section~\ref{subsec:MI-MPCprobFor} is nonlinear. The implemented online formulation requires convexification or approximation of the binary--continuous products in \eqref{eq:cellCurrent} and \eqref{eq:engagedCellVoltage}, the nonlinear OCV relation in \eqref{eq:cellVoltage}, the quadratic current terms in \eqref{eq:qi_gen}, and the inverse pack-voltage term in \eqref{eq:lambdaDef}. 
            
            To obtain a tractable mixed-integer convex approximation, the binary--continuous products are reformulated using auxiliary variables and bound-based linear constraints~\cite{bemporad1999control,williams2013model}. In particular, the product \(S_i i_{\mathrm{pack}}\) in~\eqref{eq:cellCurrent} is replaced by the cell-current variable \(i_i\), with linear constraints enforcing \(i_i=0\) when \(S_i=0\) and \(i_i=i_{\mathrm{pack}}\) when \(S_i=1\). The same construction is used for the engaged-voltage product in~\eqref{eq:engagedCellVoltage}. These reformulations are exact for valid bounds on the continuous variables.

            The squared pack current in the thermal model is represented by an epigraph variable and a second-order-cone constraint; the polarization-loss term is neglected in the prediction model, as detailed in Appendix~\ref{app:micp}. At each sampling instant, the piecewise-linear OCV map is locally replaced by the affine segment containing the SOC, with its slope and intercept held fixed over the prediction horizon. The reciprocal pack voltage in the curtailment model is approximated by a special ordered set of type 2 (SOS2) piecewise-linear interpolation over the admissible pack-voltage interval \([v_{\mathrm{pack}}^{\min},v_{\mathrm{pack}}^{\max}]\).

            As a result, we establish a mixed-integer convex approximation of the nonlinear MI-MPC formulation in~\eqref{eq:unrestrictedMIMPC}. For fixed binary engagement variables and SOS2 interpolation regions, the remaining continuous problem is convex, so branch-and-bound resolves only the engagement decisions and SOS2 adjacency choices.

            Throughout this paper, references to solving Problem~\eqref{eq:unrestrictedMIMPC}, or a ranked-prefix restriction of it, mean solving the corresponding mixed-integer convex reformulation.  Accordingly, \(J_{\mathrm{full}}^\star(k)\) in \eqref{eq:unrestrictedMIMPC_obj} denotes the optimal value of this mixed-integer convex reformulation of Problem~\eqref{eq:unrestrictedMIMPC}.
            
            The algebraic reformulations needed to specify the optimization problem are given in Appendix~\ref{app:micp}; calibration data and nonessential implementation constants are provided in~\cite{skegroRankedPrefixArchive}.

    \section{Combinatorial Structure of Cell Engagement}
    \label{sec:sourceRedundancy}

        At each prediction step \(\ell\), the full subset-selection formulation~\eqref{eq:unrestrictedMIMPC} selects a binary engagement vector \(S(\ell)\in\{0,1\}^{N}\). Certification of the resulting mixed-integer problem must therefore account for the associated binary candidate domain. This section characterizes that domain and introduces the stage-dependent cardinality floor used in~\eqref{eq:unrestrictedMIMPC}.
        
        A primary descriptor of an engagement vector is its cardinality \(m(\ell)\), introduced in~\eqref{eq:commonCardinalityFloor}, since it strongly influences the attainable pack voltage and the current required to serve a given power request. For a fixed cardinality \(m_{\mathrm{fix}}\), define the same-cardinality class
        \begin{subequations}
            \begin{align}
            \mathcal{S}_{m_{\mathrm{fix}}}
            &:=
            \{
            S\in\{0,1\}^{N}:
            \mathbf{1}^{\top}S=m_{\mathrm{fix}}
            \}, \label{eq:samecard_class}
            \\
            |\mathcal{S}_{m_{\mathrm{fix}}}|
            &=
            \binom{N}{m_{\mathrm{fix}}}.
            \end{align}
        \end{subequations}
        Fixing the number of engaged cells therefore does not remove the subset-selection problem. If \(m_{\mathrm{fix}}\) scales with the pack size (formally, if \(\varepsilon N \le m_{\mathrm{fix}} \le (1-\varepsilon)N\) for some fixed \(\varepsilon\in(0,\tfrac{1}{2}]\)), then \(\binom{N}{m_{\mathrm{fix}}}\) grows exponentially with \(N\); for an \(m_{\mathrm{fix}}\) held fixed independently of \(N\), it grows only polynomially. In either case, the multiplicity of subsets within the cardinality classes, rather than the number of cardinality values alone, accounts for the size of the binary domain.

        Before solving the MPC problem, we screen out low-cardinality engagement vectors using idealized voltage and power-capability considerations, evaluated separately at each prediction step. Deriving an admissible series-cell count from a required voltage range is an established device for reconfigurable packs~\cite{geng2025fundamental}; here it is combined with a power-capability screen into a single floor. The voltage-based threshold \(m_{\min}^{(V)}\) is computed from the pack-voltage floor and the cell-voltage ceiling as 
        \begin{equation}
            m_{\min}^{(V)}
            :=
            \left\lceil
            \frac{v_{\mathrm{pack}}^{\min}}
            {v_{\mathrm{cell}}^{\max}}
            \right\rceil .
            \label{eq:m_min_V}
        \end{equation}

        The power-capability threshold additionally requires a current ceiling consistent with the direction of the request. Denote this ceiling by \(i_{\mathrm{hw}}(\ell)\): it equals the discharge limit \(i_{\mathrm{pack}}^{\max}\) when \(p_{\mathrm{req}}(\ell)\ge 0\), and the charge limit \(|i_{\mathrm{pack}}^{\min}|\) otherwise,
        \begin{equation}
            i_{\mathrm{hw}}(\ell)
            =
            \begin{cases}
            i_{\mathrm{pack}}^{\max}, & p_{\mathrm{req}}(\ell)\ge 0,\\
            |i_{\mathrm{pack}}^{\min}|, & p_{\mathrm{req}}(\ell)<0 .
            \end{cases}
            \label{eq:iHW}
        \end{equation}
        
        Under the same idealization as \(m_{\min}^{(V)}\), and neglecting resistive losses, the full-request power-capability threshold \(m_{\min}^{(P)}(\ell)\) is the number of cells needed to serve the magnitude of \(p_{\mathrm{req}}(\ell)\) without exceeding \(i_{\mathrm{hw}}(\ell)\), each assumed to operate at \(v_{\mathrm{cell}}^{\max}\):
        \begin{equation}
            m_{\min}^{(P)}(\ell)
            :=
            \left\lceil
            \frac{|p_{\mathrm{req}}(\ell)|}
            {v_{\mathrm{cell}}^{\max} i_{\mathrm{hw}}(\ell)}
            \right\rceil .
            \label{eq:m_min_P}
        \end{equation}
        Because \(\lambda(\ell)>0\) is permitted, \(m_{\min}^{(P)}(\ell)\) is not a necessary feasibility condition for the curtailed problem: it may exclude a lower-cardinality vector that would in fact be feasible through curtailment.
        
        The two thresholds are combined and capped at the pack size,
        \begin{equation}
            m_{\min}(\ell)
            =
            \min\left(
            \max\{
            m_{\min}^{(V)}\,,
            m_{\min}^{(P)}(\ell)
            \},
            N
            \right).
            \label{eq:m_min}
        \end{equation}
        The cap only bounds the floor by the pack size. It does not assert that engaging all \(N\) cells suffices to meet the request, which remains subject to the full constraint set and, where needed, to curtailment.

        Accordingly, the resulting binary candidate domain at prediction step \(\ell\) is
        \begin{subequations}
            \begin{align}
                \mathcal{S}_{\mathrm{cand}}(\ell)
                &=
                \{
                S\in\{0,1\}^{N}:
                \mathbf{1}^{\top}S
                \ge
                m_{\min}(\ell)
                \},
                \\
                \left|
                \mathcal{S}_{\mathrm{cand}}(\ell)
                \right|
                &=
                \sum_{p=m_{\min}(\ell)}^{N}
                \binom{N}{p}.
            \end{align}
        \end{subequations}

        We call \(\mathcal{S}_{\mathrm{cand}}(\ell)\) a candidate domain, not a feasible set: it contains the binary vectors not excluded by the cardinality floor alone, before determining whether each vector admits continuous state, input, and output trajectories satisfying~\eqref{eq:unrestrictedMIMPC}. Its size quantifies the combinatorial multiplicity associated with cell engagement in the full subset-selection formulation. The admissible cardinality itself takes only \(N-m_{\min}(\ell)+1\) values, whereas each cardinality \(p\) contributes \(\binom{N}{p}\) distinct vectors to \(\mathcal{S}_{\mathrm{cand}}(\ell)\).

    \section{Ranking-Based Feasible-Set Restriction}
    \label{sec:restriction}

        We propose RFSR-MI-MPC, a ranking-based feasible-set restriction that reduces the combinatorial search space of the MI-MPC problem in Section~\ref{sec:modelAndProblemFormulation} by retaining only ranked-prefix engagement subsets. This section defines the objective-informed cell-suitability score and the induced ordering it produces, derives the resulting monotone prefix constraint and restricted MICP formulation, establishes the feasibility and optimality properties of the restriction, and presents its online implementation as a deployable algorithm.

        \subsection{Objective-Informed Cell Ranking}
        \label{subsec:objectiveinformed_cell_ranking}
            At sampling instant \(k\), the cells are ordered from the current cell states and parameters, and this ordering is held fixed over the prediction horizon. The ranking is not a greedy controller; it only defines the admissible prefix-structured binary set used by the finite-horizon MICP. For active power requests, \(p_{\mathrm{req}}(k)\neq0\), the requested-power direction \(d_p(k)\) is given by \eqref{eq:powerDirection}. Each cell is assigned a scalar suitability score
            \begin{equation}\label{eq:score_raw}
                \rho_i(k)
                =
                \rho_i^{(\mathrm{SOH})}(k)
                +
                \rho_i^{(\mathrm{SOC})}(k)
                +
                \rho_i^{(v)}(k).
            \end{equation}
            Larger values of \(\rho_i(k)\) indicate greater suitability for engagement under the current operating condition. The score is a one-step, objective-informed surrogate rather than the gradient of the exact finite-horizon mixed-integer value function.

            \paragraph{Representative Current and Terminal-Voltage Estimate}
                The score components are evaluated at a representative current magnitude
                \begin{equation}\label{eq:iexp_rank}
                    i_{\mathrm{exp}}(k)
                    =
                    \min\left(
                    \frac{|p_{\mathrm{req}}(k)|}
                    {v_{\mathrm{OC,ref}}(k)\,m_{\min}(k)},
                    \;
                    i_{\mathrm{hw}}(k)
                    \right),
                \end{equation}
                where
                \begin{equation}\label{eq:vocref_rank}
                    v_{\mathrm{OC,ref}}(k)
                    =
                    \begin{cases}
                        \min_i v_{\mathrm{OC},i}(k), & p_{\mathrm{req}}(k)>0,\\
                        \max_i v_{\mathrm{OC},i}(k), & p_{\mathrm{req}}(k)<0,
                    \end{cases}
                \end{equation}
                and \(i_{\mathrm{hw}}(k)\) is defined in \eqref{eq:iHW}. The value \(i_{\mathrm{exp}}(k)\) is used only for ranking; it is not imposed as the pack current in the MICP. The corresponding terminal-voltage estimate for cell \(i\) is
                \begin{equation}\label{eq:vtilde_rank}
                    \hat{v}_i(k)
                    =
                    v_{\mathrm{OC},i}(k)
                    -
                    d_p(k) R_0 i_{\mathrm{exp}}(k)
                    -
                    R_1 i_{\mathrm{RC},i}(k),
                \end{equation}
                with the \(R_0\) term changing sign between discharge and charge, while the RC polarization term keeps the sign of the estimated state \(i_{\mathrm{RC},i}(k)\). The estimate is clipped to the admissible voltage interval,
                \begin{equation}\label{eq:vtilde_clip_rank}
                    \hat{v}_i(k)
                    \leftarrow
                    \max\left(
                    v_{\mathrm{cell}}^{\min},
                    \min\left(\hat{v}_i(k),v_{\mathrm{cell}}^{\max}\right)
                    \right).
                \end{equation}

            \paragraph{SOH Component}
                The SOH component follows from the SOH engagement term in the horizon-averaged weighted objective:
                \begin{equation}\label{eq:rho_SOH}
                    \rho_i^{(\mathrm{SOH})}(k)
                    =
                    -
                    \frac{
                    w_{\mathrm{SOH}}\,\mathrm{SOH}^{\mathrm{EOL}}
                    }{
                    N_p N\,\mathrm{SOH}_i(k)
                    } .
                \end{equation}
                Cells with lower \(\mathrm{SOH}_i(k)\) therefore receive more negative scores and are ranked lower, all else equal.

            \paragraph{SOC-Balancing Component}
                The SOC component uses a one-step approximation of the effect of engagement on the SOC-balancing term. At the representative current,
                \begin{equation}
                    \mathrm{SOC}_i(k+1)
                    =
                    \mathrm{SOC}_i(k)
                    -
                    d_p(k) 
                    \frac{\Delta t}{3600Q_i(k)} 
                    i_{\mathrm{exp}}(k)S_i(k),
                \end{equation}
                where \(d_p(k)\) is defined in \eqref{eq:powerDirection}, and \(Q_i(k)=Q^{\mathrm{nom}}\mathrm{SOH}_i(k)\) is the effective cell capacity. Differentiating this one-step contribution to the horizon-averaged weighted objective gives
                \begin{equation}\label{eq:ci_rank}
                    c_i(k)
                    =
                    \frac{
                    2w_{\mathrm{SOC}}\Delta t\,i_{\mathrm{exp}}(k)
                    }{
                    3600 N_p N Q_i(k)\bigl(\Delta \mathrm{SOC}^{\max}(k)\bigr)^2
                    } .
                \end{equation}
                The SOC score is then
                \begin{equation}\label{eq:rho_SOC}
                    \rho_i^{(\mathrm{SOC})}(k)
                    =
                    d_p(k) c_i(k)
                    \bigl(\mathrm{SOC}_i(k)-\bar{\mathrm{SOC}}(k)\bigr).
                \end{equation}
                Thus, during discharge, above-average-SOC cells are favored; during charge, below-average-SOC cells are favored.

            \paragraph{Voltage-Curtailment Component}
                The voltage component reflects the curtailment normalization used in \(J_\lambda\), evaluated at the first prediction step. Let \(i_{\mathrm{scale}}(k)\) denote the operating-mode-dependent current normalization at sampling instant \(k\), defined in \eqref{eq:iscale} as \(i_{\mathrm{pack}}^{\max}\) during discharge and \(|i_{\mathrm{pack}}^{\min}|\) during charge. The voltage score is
                \begin{equation}
                    \label{eq:rho_v}
                    \rho_i^{(v)}(k)
                    =
                    d_p(k)
                    \frac{
                    w_\lambda i_{\mathrm{exp}}^2(k)
                    }{
                    N_p i_{\mathrm{scale}}(k)
                    \lvert p_{\mathrm{req}}(k)\rvert
                    }
                    \hat v_i(k).
                \end{equation}
                Thus, higher estimated terminal voltage is favored during discharge and lower estimated terminal voltage during charge. For \(p_{\mathrm{req}}(k)=0\), set \(\rho_i^{(v)}(k)=0\).

            \paragraph{Composite Score and Cell Ordering}
                Collect the per-cell suitability scores from \eqref{eq:score_raw} into the vector 
                \begin{equation}\label{eq:rho_vector}
                    \boldsymbol{\rho}(k) = \big(\rho_1(k),\dots,\rho_N(k)\big)^{\top}.                     
                \end{equation}
                
                The ordering permutation is defined by sorting the composite scores in non-increasing order:
                \begin{equation}\label{eq:argsort}
                    \pi_k = \operatorname{argsort}\big(-\boldsymbol{\rho}(k)\big),
                \end{equation}
                where $\pi_k : \{1,\dots,N\} \to \{1,\dots,N\}$ is the permutation such that $\pi_k(j)$ is the physical index of the cell ranked $j^{\mathrm{th}}$ by composite score, so that
                \begin{equation}
                    \rho_{\pi_k(1)}(k) \ge \rho_{\pi_k(2)}(k) \ge \cdots \ge \rho_{\pi_k(N)}(k).
                \end{equation}
                Ties are resolved by preserving the physical cell-index order, yielding a deterministic ordering for equal scores. After \(\pi_k\) is fixed, all cell-indexed variables, parameters, and constraints are reordered accordingly. The resulting ordering in \eqref{eq:argsort} is a one-step surrogate used only to construct the prefix-structured binary set used by the finite-horizon MICP. 

                \begin{proposition}[Optimality of the ranked prefix for the affine one-step surrogate]
                    \label{prop:surrogate_prefix_optimality}
                    Fix a sampling instant \(k\) and an engagement cardinality \(m_{\mathrm{fix}} \in\{0,1,\dots,N\}\). Suppose the one-step surrogate cost \(\hat{J}_{m_{\mathrm{fix}}}(s;k)\) for a binary vector \(s\in\{0,1\}^N\) satisfying \(\sum_{i=1}^N s_i=m_{\mathrm{fix}}\) can be written as
                    \begin{equation}
                        \hat{J}_{m_{\mathrm{fix}}}(s;k)
                        =
                        \beta_{m_{\mathrm{fix}}}(k)
                        -
                        \sum_{i=1}^{N}\rho_i(k)s_i ,
                    \end{equation}
                    where \(\beta_{m_{\mathrm{fix}}}(k)\) is independent of the selected subset. Under the ordering \(\pi_k\) in~\eqref{eq:argsort}, the ranked-prefix vector \(\mathbf{s}^\star \in \{0,1\}^N\), defined component-wise at each physical cell index \(\pi_k(j)\) by
                    \begin{equation}
                        s^\star_{\pi_k(j)} =
                        \begin{cases}
                            1, & j \le m_{\mathrm{fix}},\\
                            0, & j > m_{\mathrm{fix}},
                        \end{cases}
                        \qquad j = 1,\dots,N,
                    \end{equation}
                    minimizes \(\hat{J}_{m_{\mathrm{fix}}}(s;k)\) over all binary vectors of cardinality \(m_{\mathrm{fix}}\). For \(0<m_{\mathrm{fix}}<N\), this minimizer is unique among all cardinality-\(m_{\mathrm{fix}}\) binary vectors if
                    \begin{equation}
                        \rho_{\pi_k(m_{\mathrm{fix}})}(k)
                        >
                        \rho_{\pi_k(m_{\mathrm{fix}}+1)}(k).
                    \end{equation}
                \end{proposition}
                
                This is the classical greedy optimality theorem for weighted matroids, specialized here to the uniform matroid (equivalently, top-selection); see~\cite{edmonds1971matroids, oxley1992matroid}. For completeness, we include a short exchange proof.
                \begin{proof}
                    Since \(\beta_{m_{\mathrm{fix}}}(k)\) is independent of the selected subset, minimizing \(\hat{J}_{m_{\mathrm{fix}}}(s;k)\) is equivalent to maximizing
                    \[
                        \sum_{i=1}^{N}\rho_i(k)s_i
                    \]
                    subject to \(s_i\in\{0,1\}\) and \(\sum_i s_i=m_{\mathrm{fix}}\). If a feasible selection contains a cell \(\pi_k(q)\) with \(q>m_{\mathrm{fix}}\) and omits a cell \(\pi_k(p)\) with \(p\le m_{\mathrm{fix}}\), exchanging these two selections changes the score sum by
                    \[
                        \rho_{\pi_k(p)}(k)-\rho_{\pi_k(q)}(k)\ge0 .
                    \]
                    Thus the exchange cannot increase the surrogate cost. Repeating this exchange yields the ranked-prefix vector. If \(\rho_{\pi_k(m_{\mathrm{fix}})}(k)>\rho_{\pi_k(m_{\mathrm{fix}}+1)}(k)\), every non-prefix cardinality-\(m_{\mathrm{fix}}\) vector omits at least one strictly higher-scored prefix cell and includes at least one strictly lower-scored non-prefix cell, so the ranked-prefix minimizer is unique.
                \end{proof}

        \subsection{Monotone Prefix Constraint}
        \label{subsec:monotone_prefix}

            At sampling instant \(k\), the score ordering \(\pi_k\) is fixed before the finite-horizon MICP is solved. We express the binary engagement variables in this ordered coordinate system as
            \begin{equation}\label{eq:reordered_coords}
                \tilde S_j(\ell) = S_{\pi_k(j)}(\ell),
                \quad j\in\mathcal C,\quad \ell\in\mathcal H_k .
            \end{equation}           
            Note that $\tilde{S}_j(\ell)$ is defined by evaluating the original engagement variable $S_i(\ell)$ at the physical cell index $i = \pi_k(j)$; the rank $j$ enters only as the argument of $\pi_k(\cdot)$, which returns that physical index. The ranked-prefix restriction is imposed by the linear monotonicity constraints
            \begin{equation}\label{eq:monotone}
                \tilde S_1(\ell)
                \ge
                \tilde S_2(\ell)
                \ge
                \cdots
                \ge
                \tilde S_N(\ell),
                \qquad
                \tilde S_j(\ell)\in\{0,1\},                
            \end{equation}
            for every prediction step \(\ell\in\mathcal H_k\).
            
            Because the variables are binary, \eqref{eq:monotone} admits exactly the prefix vectors
            \begin{equation}\label{eq:prefix_shape}
                \tilde S^{(m)}
                =
                \begin{bmatrix}
                \underbrace{1,\ldots,1}_{m},
                \underbrace{0,\ldots,0}_{N-m}
                \end{bmatrix}^{\top},
                \qquad m=0,\ldots,N .
            \end{equation}
            Thus, for a given prefix length \(m\), the engaged cells are precisely the \(m\) highest-ranked cells under \(\pi_k\). The optimizer does not choose an arbitrary subset of \(m\) cells; it chooses only the prefix length.
            
            Because the common floor \eqref{eq:commonCardinalityFloor} is already part of the full subset-selection problem, imposing \eqref{eq:monotone} leaves the admissible prefix lengths
            \begin{equation}
                m(\ell)
                =
                \sum_{j=1}^{N}\tilde S_j(\ell)
                \in
                \{m_{\min}(\ell),\ldots,N\},
            \label{eq:m_select}
            \end{equation}
            where \(m(\ell)\) is the same engagement cardinality defined via \eqref{eq:commonCardinalityFloor}: since \(\pi_k\) is a permutation of the cell indices, the reordering~\eqref{eq:reordered_coords} leaves the sum \(\sum_i S_i(\ell)\) unchanged, so \(\sum_{j=1}^N \tilde S_j(\ell) = \sum_{i=1}^N S_i(\ell)\).
            
            Therefore, at prediction step \(\ell\), the number of admissible binary engagement patterns is reduced from
            \begin{equation}
                \sum_{p=m_{\min}(\ell)}^{N}\binom{N}{p}
            \end{equation}
            to
            \begin{equation}
                N-m_{\min}(\ell)+1 .
            \end{equation}
            
            The restriction adds no binary variables. It only adds the \(N-1\) linear inequalities in~\eqref{eq:monotone} per prediction step, after reordering the cell-indexed quantities according to \(\pi_k\). 

            Collecting~\eqref{eq:reordered_coords} and~\eqref{eq:monotone} together with the original finite-horizon problem~\eqref{eq:unrestrictedMIMPC}, the ranked-prefix restricted problem solved at sampling instant $k$ is        
            \begin{equation}
                \label{eq:restrictedMIMPC}
                \begin{aligned}
                    J_{\mathrm{pref}}^\star(k)
                    &=
                    \min_{\mathcal{U}_k}\; J(k)\\
                    \mathrm{s.t.}\quad
                    &\text{all constraints of~\eqref{eq:unrestrictedMIMPC}},\\
                    &\tilde S_j(\ell)=S_{\pi_k(j)}(\ell),
                    \qquad j\in\mathcal C,\ \ell\in\mathcal H_k,\\
                    &\tilde S_j(\ell)\ge\tilde S_{j+1}(\ell),
                    \qquad j=1,\ldots,N-1,\ \ell\in\mathcal H_k.
                \end{aligned}
            \end{equation}
            Thus, Problem~\eqref{eq:restrictedMIMPC} is the mixed-integer convex reformulation of Problem~\eqref{eq:unrestrictedMIMPC} (Section~\ref{subsec:micp}), with the prefix inequalities~\eqref{eq:monotone} added; $J_{\mathrm{pref}}^\star(k)$ denotes its optimal value, consistently with $J_{\mathrm{full}}^\star(k)$ in~\eqref{eq:unrestrictedMIMPC_obj}. This restricted formulation, together with the cell-ranking procedure of Section~\ref{subsec:objectiveinformed_cell_ranking}, constitutes RFSR-MI-MPC. Because the cardinality floor~\eqref{eq:commonCardinalityFloor}, together with the thresholds defined in Section~\ref{sec:sourceRedundancy}, carries over unchanged from \eqref{eq:unrestrictedMIMPC} into \eqref{eq:restrictedMIMPC}, this floor is identical for the full subset-selection and ranked-prefix formulations and therefore does not confound their comparison.

        \subsection{Feasibility and Optimality}
        \label{sec:properties}

            The ranked-prefix restriction modifies only the admissible binary engagement set after the common floor has been imposed. Let \(\mathcal{Z}_{\mathrm{full}}(k)\) denote the feasible set of the full subset-selection MICP, including \eqref{eq:commonCardinalityFloor}, and let \(\mathcal{Z}_{\mathrm{pref}}(k)\) denote the feasible set obtained by additionally imposing \eqref{eq:monotone}. Hence,
            \begin{equation}
                \mathcal{Z}_{\mathrm{pref}}(k)
                \subseteq
                \mathcal{Z}_{\mathrm{full}}(k),
                \label{eq:feasset_inclusion}
            \end{equation}
            and, whenever both problems are feasible,
            \begin{equation}
                J_{\mathrm{full}}^\star(k)
                \le
                J_{\mathrm{pref}}^\star(k).
                \label{eq:restricted_value_upper}
            \end{equation}          
    
            Thus, optimality certificates obtained for the ranked-prefix formulation apply to \(\mathcal{Z}_{\mathrm{pref}}(k)\), not necessarily to \(\mathcal{Z}_{\mathrm{full}}(k)\). The all-engaged vector \(\mathbf{1}_N\) is a prefix under every ordering, so any feasible all-engaged finite-horizon trajectory remains feasible after restriction; this is only a sufficient fallback candidate, not a recursive-feasibility guarantee.
            
            \begin{proposition}[Finite-horizon exactness]
                \label{prop:exactness}
                Fix a sampling instant $k$ and use the ordering $\pi_k$ in~\eqref{eq:argsort}. If the full subset-selection problem~\eqref{eq:unrestrictedMIMPC} admits a globally optimal binary trajectory \(\{S^\star(\ell)\}_{\ell\in\mathcal H_k}\) that satisfies the ranked-prefix constraint~\eqref{eq:monotone} for every $\ell\in\mathcal H_k$, then
               \begin{equation}
                    J_{\mathrm{pref}}^\star(k)
                    =
                    J_{\mathrm{full}}^\star(k).
                \end{equation}
            \end{proposition}          

            Set inclusion gives \(J_{\mathrm{pref}}^\star(k)\ge J_{\mathrm{full}}^\star(k)\). The hypothesized full-subset optimizer is feasible for the ranked-prefix problem, which gives the reverse inequality.

            When this condition fails, the ranked-prefix solution can be suboptimal relative to the full subset-selection finite-horizon optimum. No general finite-horizon sub-optimality bound is claimed; the approximation cost is instead quantified empirically using same-state full-subset reference solves.

        \subsection{Online Implementation}
        \label{subsec:onlineImplementation}
            
            The online implementation of RFSR-MI-MPC is summarized in Algorithm~\ref{alg:onlineRankedPrefixMI-MPC}. The computational overhead introduced by the ranking procedure consists of \(O(N)\) operations for score evaluation, \(O(N\log N)\) time for sorting, and \(O(N)\) memory for storing the score and permutation vectors. The resulting MICP is solved directly using a standard solver without callbacks, custom branching rules, offline training, or solver-specific modifications.
            
            \begin{algorithm}[!ht]
                \caption{Online RFSR-MI-MPC}
                \label{alg:onlineRankedPrefixMI-MPC}
                \begin{algorithmic}[1]
                    \Require Cell states; requested-power sequence
                    \({p_{\mathrm{req}}(\ell)}\,,{\ell\in\mathcal H_k}\); previous bound
                    \(\Delta \mathrm{SOC}^{\max}(k-1)\).
                    \Ensure Applied engagement vector \(S(k)\) and pack-current reference
                    \(i_{\mathrm{pack}}(k)\).
                    
                        \State Initialize the predicted trajectory from the current state \(x(k)\); update \(\Delta \mathrm{SOC}^{\max}(k)\) using \eqref{eq:dSOCmaxTarget}--\eqref{eq:dSOCmaxSlew}; and compute \(\{m_{\min}(\ell)\}_{\ell\in\mathcal H_k}\) using \eqref{eq:m_min_V}--\eqref{eq:m_min}.
                    
                        \State Compute the representative current and terminal-voltage estimates
                        using \eqref{eq:iexp_rank}--\eqref{eq:vtilde_clip_rank}; for
                        \(p_{\mathrm{req}}(k)=0\), set \(i_{\mathrm{exp}}(k)=0\).
                    
                        \State Evaluate the composite scores \(\{\rho_i(k)\}_{i\in\mathcal C}\)
                        from \eqref{eq:score_raw}--\eqref{eq:rho_v} and obtain the deterministic
                        ordering \(\pi_k\) from \eqref{eq:argsort}.
                    
                        \State Reorder the cell-indexed quantities according to \(\pi_k\) and
                        impose the ranked-prefix constraints
                        \eqref{eq:reordered_coords}--\eqref{eq:m_select} over the horizon.
                    
                        \State Solve the restricted MICP \eqref{eq:restrictedMIMPC}.
                    
                        \State Apply the first action of the returned feasible solution; if no
                        feasible solution is returned, apply the prescribed zero-command fallback.
                \end{algorithmic}
            
            \end{algorithm}

    \section{Numerical Evaluation Setup}
    \label{sec:evalMethodology}

        We consider a nominal \(N=20\) reconfigurable pack, corresponding to the configuration in Fig.~\ref{fig:pack_topology}. Hereafter, RFSR-MI-MPC is referred to as the proposed controller, whereas the full subset-selection implementation of~\eqref{eq:unrestrictedMIMPC} is referred to as the baseline controller.

        \subsection{Simulation, Controller, and Solver Settings}
        \label{subsec:evaluation_protocol_label}

            In the closed-loop simulations, the battery-pack states are propagated using the electrothermal model in Section~\ref{subsec:electhermBattModel}, and both controllers solve the mixed-integer convex reformulation described in Section~\ref{subsec:micp}. The baseline and proposed controllers are evaluated under identical numerical and solver settings. Specifically, both controllers are implemented in MATLAB using YALMIP, and the resulting optimization problems are solved using Gurobi 13.0.1. Table~\ref{tab:evaluation_protocol} summarizes the nominal simulation, controller, and solver parameters. Complete model-calibration files, OCV data, the processed demand trace, random seeds, and reproduction scripts are provided in~\cite{skegroRankedPrefixArchive}; the simulation result files from which every reported number and figure is regenerated are archived in~\cite{skegroRankedPrefixData}.

            \begin{table}[ht]
                \centering
                \scriptsize
                \caption{Nominal Evaluation Parameters.}
                \label{tab:evaluation_protocol}
                \setlength{\tabcolsep}{2.0pt}
                \begin{tabular}{@{}lclc@{}}
                    \toprule
                    Quantity & Value & Quantity & Value \\
                    \midrule
                    \(N\) & 20
                    & \(N_p\) & 10 \\
                    \(\Delta t\) & \(1~\mathrm{s}\)
                    & Demand & Scaled WLTC Class~3 \\
                    \(\bigl(w_{\mathrm{SOH}},w_{\mathrm{SOC}},w_\lambda\bigr)\)
                    & \((10,4,1000)\)
                    & \(t_{\mathrm{lim}}^{\mathrm{op}}\) & \(1~\mathrm{s}\) \\
                    Relative MIP-gap tolerance & \(10^{-4}\)
                    & CPU & i7-1370P (13th Gen) \\
                    \(v_{\mathrm{cell}}\) & \([2.0,3.60]~\mathrm{V}\)
                    & \(v_{\mathrm{pack}}\) & \([44.8,72.0]~\mathrm{V}\) \\
                    \(i_{\mathrm{pack}}\) & \([-3.0,18.0]~\mathrm{A}\)
                    & \(\mathrm{SOC}_{\mathrm{cell}}\) & \([0.05,1.00]\) \\
                    \(T_{\mathrm{cell}}\) & \([11,44]~{}^{\circ}\mathrm{C}\)
                    & \(\mathrm{SOH}^{\mathrm{EOL}}\) & 0.80 \\
                    \(\delta_P\) & 0.01
                    & \((\Delta_{\mathrm{low}},\Delta_{\mathrm{high}})\)
                    & \((0.03,0.06)\) \\
                    \((z_1,z_2)\) & \((0.25,0.40)\)
                    & \(\sigma_\Delta\) & \(0.005~\mathrm{p.u.}/\mathrm{s}\) \\
                    \(\mathrm{SOC}_i(0)\) & \(\mathcal{N}(0.8,0.001^2)\)
                    & \(\mathrm{SOH}_i(0)\) & \(\mathcal{N}(0.98,0.01^2)\) \\
                    \(T_i(0)\) & \(25~{}^{\circ}\mathrm{C}\)
                    & \(i_{\mathrm{RC},i}(0)\) & \(0~\mathrm{A}\) \\
                    \multicolumn{4}{@{}p{0.47\textwidth}@{}}{\footnotesize $\mathcal{N}(\mu,\sigma^2)$ denotes a Gaussian distribution with mean $\mu$ and variance $\sigma^2$, clipped to the admissible SOC/SOH range.} \\
                    \bottomrule
                \end{tabular}
            \end{table}

            The pack-power request is generated from the WLTC Class~3b speed trace~\cite{UN_GTR15} using a backward-facing quasi-static longitudinal vehicle model~\cite{guzzella2013vehicle}. The model accounts for rolling resistance, aerodynamic drag, vehicle and rotational inertia, drivetrain and regenerative-braking efficiencies, and a constant \(300~\mathrm{W}\) auxiliary load at the unscaled vehicle level. The resulting battery-power trace, including the auxiliary load, is uniformly scaled by \(\kappa_{\mathrm{WLTC}}=0.008122\), yielding a maximum discharge request of \(0.444~\mathrm{kW}\) and a maximum regenerative-charge magnitude of \(0.192~\mathrm{kW}\). The vehicle parameters, processed speed trace, and demand-generation implementation are provided in~\cite{skegroRankedPrefixArchive}.

        \subsection{Reporting Conventions and Certification Metrics}
        \label{subsec:reportingConventions}

            With the simulation environment and solver settings fixed, we next define the reporting metrics. The set of active control steps \(\mathcal K_{\mathrm{act}}\) and its cardinality \(K_{\mathrm{act}}\) are defined as
            \begin{equation}
                \mathcal K_{\mathrm{act}}
                =
                \{k:p_{\mathrm{req}}(k)\neq0\},
                \qquad
                K_{\mathrm{act}}=|\mathcal K_{\mathrm{act}}|.
            \end{equation}
            Unless otherwise stated, all reported statistics are computed over $\mathcal K_{\mathrm{act}}$. For any per-step scalar quantity $X$, $\bar X$ denotes its sample mean, $X_{\alpha}$ its sample $\alpha^{\mathrm{th}}$ percentile, and $X_{\max}$ its maximum over the stated evaluation set; in particular, $X_{50}$ is the sample median. The curtailment rate \(\Phi_{\mathrm{curt}}\) is the percentage of active steps at which the applied solution curtails the power request,
            \begin{equation}
                \Phi_{\mathrm{curt}}
                =
                \frac{100}{K_{\mathrm{act}}}
                \bigl|\{k\in\mathcal K_{\mathrm{act}}:\lambda(k)>\varepsilon_\lambda\}\bigr|
                \quad[\%],
                \label{eq:curtRate}
            \end{equation}
            where $\lambda(k)$ is the applied curtailment slack from~\eqref{eq:lambdaDef} and $\varepsilon_\lambda=0.01~\mathrm{A}$ is the detection threshold.

            To assess solver reliability, for each active step \(k\in\mathcal K_{\mathrm{act}}\) we record whether a feasible action is obtained and whether optimality is certified within the operational time limit \(t_{\mathrm{lim}}^{\mathrm{op}}=1~\mathrm{s}\). We write \(t(k)\) for the solver time and \(n^{\mathrm{node}}(k)\) for the number of branch-and-bound nodes. A step is counted as \emph{root-node} when the solver reports \(n^{\mathrm{node}}(k)\le 1\), that is, when the instance is closed in presolve or at the root relaxation without branching. A step is classified as \textit{Optimal} when the solver certifies optimality to the prescribed relative MIP-gap tolerance of \(10^{-4}\) within the allotted time, \textit{TO+Inc.} when the time limit is reached with a feasible incumbent but no certificate, and \textit{Fail} when no usable solution is returned (reported infeasibility, solver error, or a time-out without an incumbent), in which case the prescribed zero-command fallback of Algorithm~\ref{alg:onlineRankedPrefixMI-MPC} is applied. The certified-optimality rate \(\eta_{\mathrm{opt}}\) is
            \begin{equation}
                \eta_{\mathrm{opt}}
                =
                \frac{
                |\{k\in\mathcal K_{\mathrm{act}}:
                \text{step \(k\) is Optimal}\}|
                }{K_{\mathrm{act}}}.
                \label{eq:etaOptMetric}
            \end{equation}

        \subsection{Same-State Reference-Solve Protocol}
        \label{subsec:sameStateProtocol}

            A direct comparison of the closed-loop objective values of the two controllers does not isolate the objective difference introduced by the ranked-prefix restriction, because the proposed and baseline controllers generally generate different state trajectories. The restriction cost is therefore quantified empirically by solving both formulations from the same pre-control state \(x(k)\), on two independently generated state trajectories, so that the resulting cost estimate does not depend on which trajectory the states were drawn from.

            The first trajectory is generated by the proposed controller under the \(1~\mathrm{s}\) operational solver limit. At each of the 1800 active states visited along this trajectory, the full subset-selection formulation is also solved from the same state \(x(k)\) using a \(100~\mathrm{s}\) reference time limit. The second trajectory is generated independently by a deterministic rule-based policy that cyclically shifts the engagement set by one physical cell at each active discharge step, engaging 17 of the 20 cells; during active charge (regenerative) steps, the policy engages all 20 cells. This policy is memoryless and does not use SOC, SOH, or the suitability score of Section~\ref{subsec:objectiveinformed_cell_ranking}. At each recorded state \(x(k)\) on this trajectory, the ranked-prefix and full subset-selection formulations are solved independently using the \(1~\mathrm{s}\) operational and \(100~\mathrm{s}\) reference limits, respectively; neither optimization affects the trajectory itself. The second trajectory therefore provides an out-of-policy evaluation of the restriction cost, allowing an assessment of whether the performance difference is specific to states induced by the proposed controller.

            For each state sequence, let \(f\in\{\mathrm{full},\mathrm{pref}\}\) denote the formulation type and
            \begin{equation}
                \mathcal K_{\mathrm{opt}}^{f}
                :=
                \left\{
                    k\in\mathcal K_{\mathrm{act}}:
                    \text{\(f\) is classified as \textit{Optimal}}
                \right\},
                \label{eq:KoptDef}
            \end{equation}
            denote the set of active steps at which \(f\) is certified optimal. The set of jointly certified, uncurtailed active steps is then defined as
            \begin{equation}
                \begin{aligned}
                    \mathcal K_{\mathrm{pair}}
                    :=
                    \Bigl\{
                        k\in
                        \mathcal K_{\mathrm{opt}}^{\mathrm{full}}
                        \cap
                        \mathcal K_{\mathrm{opt}}^{\mathrm{pref}}
                        :\;&
                        \lambda_{\mathrm{pref}}(k)\le\varepsilon_\lambda,\\
                        &
                        \lambda_{\mathrm{full}}(k)\le\varepsilon_\lambda
                    \Bigr\}.
                \end{aligned}
                \label{eq:KpairDef}
            \end{equation}
            where \(\lambda_{\mathrm{pref}}\) and \(\lambda_{\mathrm{full}}\) are the curtailment slacks returned by the ranked-prefix and full subset-selection formulations, respectively, when solved from state \(x(k)\); both are bounded above by the detection threshold \(\varepsilon_{\lambda}\). On \(\mathcal K_{\mathrm{pair}}\), the relative restriction gap \(g_{\mathrm{rel}}\) is defined as
            \begin{equation}
                g_{\mathrm{rel}}(k)
                =
                100\,
                \frac{
                J_{\mathrm{pref}}^\star(k)-J_{\mathrm{full}}^\star(k)
                }{
                J_{\mathrm{full}}^\star(k)
                }
                \quad[\%],
                \quad
                k\in\mathcal K_{\mathrm{pair}} .
                \label{eq:g_rel_metric}
            \end{equation}

            To identify the objective terms responsible for the restriction cost, let \(c\) denote an objective component. Let \(J_c^{f,\star}(\ell;k)\) denote the stage cost \(J_c(\ell)\) evaluated along the certified optimal solution of formulation \(f\) initialized from \(x(k)\). Its weighted mean \(\bar J_{c,w}^{\mathrm{SS},f}\), evaluated over the prediction horizons of all paired states, is
            \begin{equation}
                \bar J_{c,w}^{\mathrm{SS},f}
                =
                \frac{1}{|\mathcal K_{\mathrm{pair}}|}
                \sum_{k\in\mathcal K_{\mathrm{pair}}}
                \frac{1}{N_p}
                \sum_{\ell=k+1}^{k+N_p}
                w_c J_c^{f,\star}(\ell;k),
                \label{eq:barJcwSS}
            \end{equation}
            where \(w_c\) is the corresponding objective weight in~\eqref{eq:ContrOBJ}. The restriction-induced change in component \(c\in\{\mathrm{SOH},\mathrm{SOC},\lambda\}\), denoted \(\Delta\bar J_{c,w}^{\mathrm{SS}}\), is
            \begin{equation}
                \Delta\bar J_{c,w}^{\mathrm{SS}}
                :=
                \bar J_{c,w}^{\mathrm{SS},\mathrm{pref}}
                -
                \bar J_{c,w}^{\mathrm{SS},\mathrm{full}} .
                \label{eq:deltaComponentSS}
            \end{equation}
            The mean total restriction cost \(\Delta\bar J_{\mathrm{obj}}^{\mathrm{SS}}\) is calculated by
            \begin{equation}
                \begin{aligned}
                \Delta\bar J_{\mathrm{obj}}^{\mathrm{SS}}
                &:=
                \frac{1}{|\mathcal K_{\mathrm{pair}}|}
                \sum_{k\in\mathcal K_{\mathrm{pair}}}
                \left[
                J_{\mathrm{pref}}^\star(k)
                -
                J_{\mathrm{full}}^\star(k)
                \right] \\
                &=
                \Delta\bar J_{\mathrm{SOH},w}^{\mathrm{SS}}
                +
                \Delta\bar J_{\mathrm{SOC},w}^{\mathrm{SS}}
                +
                \Delta\bar J_{\lambda,w}^{\mathrm{SS}} .
                \end{aligned}
                \label{eq:deltaBarJSS}
            \end{equation}

        \subsection{Closed-Loop Performance Metrics}
        \label{subsec:closedLoopMetrics}

            For \(k\in\mathcal K_{\mathrm{act}}\), the relative power-tracking error \(e\) is
            \begin{equation}
                e(k)
                =
                \frac{
                |p_{\mathrm{pack}}(k)-p_{\mathrm{req}}(k)|
                }{
                |p_{\mathrm{req}}(k)|
                } .
                \label{eq:trackingErrorCL}
            \end{equation}

            Balancing performance is quantified through the cell-to-cell SOC dispersion \(\sigma_{\mathrm{SOC}}\), defined as the population standard deviation of the cell SOC values. Its cycle mean \(\bar{\sigma}_{\mathrm{SOC}}\) is calculated over all simulation steps after the initial condition, including zero-power steps.

            To distinguish improved SOC regulation from simply engaging a larger fraction of the pack, we also quantify cell utilization and switching activity. The mean fraction of cells engaged during active operation \(\bar{\phi}_{\mathrm{eng}}\), calculated to characterize cell utilization, is
            \begin{equation}
                \bar{\phi}_{\mathrm{eng}}
                =
                \frac{1}{N K_{\mathrm{act}}}
                \sum_{k\in\mathcal K_{\mathrm{act}}}
                \sum_{i=1}^{N} S_i(k).
                \label{eq:phiEng}
            \end{equation}
            Switching activity at an active step is calculated by the number of cells \(n_{\mathrm{sw}}\) whose engagement state changes relative to the immediately preceding simulation step,
            \begin{equation}
                n_{\mathrm{sw}}(k)
                =
                \sum_{i=1}^{N}
                \left|S_i(k)-S_i(k-1)\right|,
                \quad k\in\mathcal K_{\mathrm{act}}\setminus\{1\}.
                \label{eq:nsw}
            \end{equation}
            Step \(k = 1\) is excluded because no prior engagement command \(S_i(0)\) exists at the start of the simulation.

            To summarize the finite-horizon solutions returned during closed-loop operation, let
            \begin{equation}
                \mathcal K_{\mathrm{usable}}
                :=
                \{k\in\mathcal K_{\mathrm{act}}:
                \text{step $k$ is \textit{Optimal} or \textit{TO+Inc.}}\}
                \label{eq:KusableDef}
            \end{equation}
            denote the set of usable closed-loop steps, that is, the active steps excluding \textit{Fail} steps, at which no weighted finite-horizon objective is defined. The mean weighted contribution of objective component \(c\) over these steps is
            \begin{equation}
                \bar J_{c,w}^{\mathrm{CL}}
                =
                \frac{1}{|\mathcal K_{\mathrm{usable}}|}
                \sum_{k\in\mathcal K_{\mathrm{usable}}}
                \frac{1}{N_p}
                \sum_{\ell=k+1}^{k+N_p}
                w_c J_c^{\mathrm{ret}}(\ell;k),
                \label{eq:barJclosedLoop}
            \end{equation}
            where \(J_c^{\mathrm{ret}}(\ell;k)\) denotes stage-cost component \(J_c\) at prediction stage \(\ell\), evaluated along the finite-horizon solution returned by the online solver at sampling instant \(k\). At an \textit{Optimal} step, the returned solution is the certified optimum; at a \textit{TO+Inc.} step, it is the feasible incumbent available when the operational time limit is reached. The corresponding mean returned objective is
            \begin{equation}
                \bar J_{\mathrm{obj}}^{\mathrm{CL}}
                =
                \bar J_{\mathrm{SOH},w}^{\mathrm{CL}}
                +
                \bar J_{\mathrm{SOC},w}^{\mathrm{CL}}
                +
                \bar J_{\lambda,w}^{\mathrm{CL}} .
                \label{eq:barJobjCL}
            \end{equation}
            Only the first action of each returned horizon solution is applied. Hence, \(\bar J_{\mathrm{obj}}^{\mathrm{CL}}\) summarizes the finite-horizon objective values returned online along a controller's closed-loop trajectory; it is not a realized cycle cost, and it does not isolate the objective loss caused by the ranked-prefix restriction. That quantity is assessed by the same-state protocol of Section~\ref{subsec:sameStateProtocol}, where both formulations are initialized from the same state and compared only when both optima are certified.

        \subsection{Monte Carlo and Robustness Sweep Design}
        \label{subsec:mcRobustnessDesign}

            \paragraph{Initial-condition Monte Carlo design}
                To assess dependence on the initial cell states, we perform \(N_{\mathrm{MC}}=20\) paired WLTC simulations with independently sampled initial SOC and SOH values from the distributions in Table~\ref{tab:evaluation_protocol}. Within each realization, the baseline and proposed controllers start from the same sampled pack state. For a performance metric \(M\), let
                \begin{equation}
                    d_h(M) = M_{\mathrm{prop},h} - M_{\mathrm{base},h},
                    \quad
                    h=1,\ldots,N_{\mathrm{MC}},
                \end{equation}
                denote the paired difference for realization \(h\). We report the median value of each metric across the 20 realizations for each controller, together with the mean paired difference
                \begin{equation}
                    \overline{\Delta M}
                    =
                    \frac{1}{N_{\mathrm{MC}}}
                    \sum_{h=1}^{N_{\mathrm{MC}}} d_h(M).
                    \label{eq:mcMeanPairedDifference}
                \end{equation}
                The corresponding \(95\%\) confidence interval is the percentile-bootstrap interval for \(\overline{\Delta M}\), obtained from \(B=10{,}000\) resamples of the 20 paired differences. Thus, the difference of the two medians need not equal \(\overline{\Delta M}\).

            \paragraph{Heterogeneity and demand sweep design}
                To test whether the restriction remains effective as the cells become less similar and as the demand departs from the WLTC profile, we evaluate four scenarios at \(N=20\) and \(N_p=10\). Table~\ref{tab:robustness_case_params} lists the SOC- and SOH-spread values used in each case.
                \begin{table}[!ht]
                \centering
                \footnotesize
                \caption{Initial-Dispersion Parameters for the Heterogeneity/Demand Sweep.}
                \label{tab:robustness_case_params}
                \begin{tabular}{lccc}
                \hline\hline
                Case & $\sigma_{\mathrm{SOC},0}$ [\%] & $\sigma_{\mathrm{SOH},0}$ [\%] & Demand \\
                \hline
                WLTC Low    & 0.05 & 0.5 & WLTC \\
                WLTC Med    & 0.10 & 1.0 & WLTC \\
                WLTC High   & 1.00 & 1.5 & WLTC \\
                High Power   & 0.10 & 0.5 & 200~W const. \\
                \hline\hline
                \end{tabular}
                \end{table}
                The three WLTC cases use the WLTC demand profile at low, medium, and high initial dispersion, respectively. The High Power case uses the medium SOC-spread level \(0.10\%\) together with the low SOH-spread level \(0.5\%\) and replaces the WLTC sequence by a constant \(200~\mathrm{W}\) discharge request. This case requires a nontrivial engagement decision at every active step and separates the effect of the restriction from the temporal variability of the driving cycle.

            \paragraph{Pack-size sweep design}
                To test whether the computational advantage persists as the number of cell-engagement decisions increases, we vary the pack size over \(N\in\{16,20,24\}\), while holding the prediction horizon, objective weights, solver settings, and initial-state distribution fixed. The WLTC power profile is rescaled for each pack size by scaling \(\kappa_{\mathrm{WLTC}}\) (Section~\ref{subsec:evaluation_protocol_label}) proportionally to \(N\), resulting in scaled peak discharge powers of \(355\), \(444\), and \(533~\mathrm{W}\) for \(N=16\), \(20\), and \(24\), respectively. This scaling keeps the per-cell C-rate and duty cycle comparable across pack sizes, so that the sweep isolates the effect of the number of binary engagement decisions under comparable per-cell operating conditions.

            \paragraph{Prediction-horizon sweep design}
                To test whether the ranked-prefix restriction changes the usual tradeoff between horizon length and online computational burden, we vary the horizon over \(N_p\in\{4,6,8,10\}\) with \(N=20\), the medium-heterogeneity WLTC case, and otherwise unchanged controller and solver settings.

    \section{Results and Discussion}
    \label{sec:results}

        This section evaluates the proposed controller in four stages. We first determine whether the ranked-prefix restriction improves real-time certification under the nominal WLTC cycle. We then quantify the approximation cost introduced by the restriction through same-state reference solves, examine how the resulting controller behaves when deployed in closed loop, and finally test whether the computational and control benefits persist across variations in initial conditions, cell heterogeneity, demand, pack size, and prediction horizon.

        \subsection{Nominal Driving-Cycle Evaluation}
        \label{sec:results_nominal}
        
            The nominal WLTC case is used to examine three distinct aspects of controller performance that should be considered separately: whether a solution can be certified within the available sampling interval, the finite-horizon optimality loss introduced by restricting the engagement set, and the resulting closed-loop behavior of the proposed controller.

            \subsubsection{Solver Certification and Runtime}
            \label{sec:vb1_runtime}
                
                We begin with the computational requirement, since real-time feasibility is a prerequisite for meaningful closed-loop deployment. Table~\ref{tab:runtime_timing} and Fig.~\ref{fig:VB1_solve_time_cdf} compare the certification and runtime behavior of the baseline and proposed controllers over the nominal WLTC cycle.
                
                \begin{table}[!ht]
                    \centering
                    \footnotesize
                    \caption{Solver-Time And Branch-and-Bound Statistics For The Nominal WLTC Cycle.}
                    \label{tab:runtime_timing}
                    \setlength{\tabcolsep}{2pt}
                    \begin{tabular}{lccccccc}
                    \hline\hline
                    Controller & $t_{50}$ [ms] & $t_{95}$ [ms] & $t_{\max}$ [ms] & Root [\%]
                         & $n^{\mathrm{node}}_{50}$ & $n^{\mathrm{node}}_{95}$ &
                         $n^{\mathrm{node}}_{\max}$ \\
                    \hline
                    Baseline & 310.5 & 1012.0 & 1289.0 &  85.0 & 1 & 289 & 1089 \\
                    Proposed & 106.0 &  148.0 &  211.0 & 100.0 & 1 &   1 &    1 \\
                    \hline
                    Speedup & $2.9{\times}$ & $6.8{\times}$ & $6.1{\times}$ & — & — & — & —\\
                    \hline\hline
                    \end{tabular}
                \end{table}                 
                
                \begin{figure}[!ht]
                    \centering
                    \includegraphics[width=\linewidth]{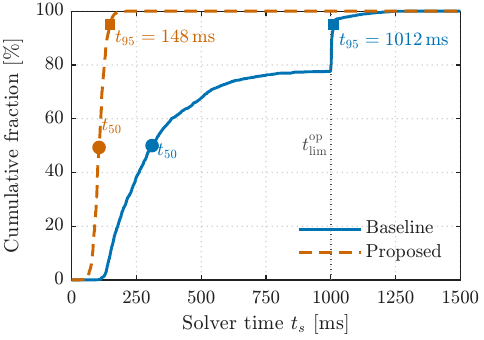}
                    \caption{Empirical cumulative distribution function of per-step solver time over the nominal WLTC cycle. The dashed line marks the operational solver time limit \(t_{\mathrm{lim}}^{\mathrm{op}}=1~\mathrm{s}\).}
                    \label{fig:VB1_solve_time_cdf}
                \end{figure}         

                The computational difference is concentrated in the upper tail of the runtime distribution. Both controllers resolve at least half of the instances at the root node, but the baseline requires branch-and-bound search on a subset of steps, reaching \(289\) nodes at the \(95^{\mathrm{th}}\) percentile and \(1089\) nodes in the worst case. In contrast, every proposed-controller instance terminates without branching. Consequently, the proposed controller reduces \(t_{95}\) by a factor of \(6.8\), compared with a \(2.9\times\) reduction in median solver time. Because the median runtime is below the operational limit for both controllers, the reduction in the upper tail is the more relevant improvement for real-time operation. Fig.~\ref{fig:VB1_solve_time_cdf} illustrates this distinction through the markedly narrower runtime distribution of the proposed controller.

                The baseline fails to certify optimality on \(22.39\%\) of the active steps, whereas the proposed controller certifies every active step. Neither controller encounters a \textit{Fail} event, and most baseline timeouts return low-gap feasible incumbents, with a median MIP gap of \(1.230\%\). Thus, the principal computational limitation of the full subset-selection formulation is not obtaining a feasible control action, but certifying its optimality within the prescribed computation time.
            
                Having established that the missed deadlines are primarily a certification problem rather than a failure to obtain feasible control, we next examine when those difficult solves arise. Non-certified baseline steps occur at substantially higher absolute power requests than certified steps, indicating that the unrestricted subset-selection problem becomes more difficult under demanding power conditions.
                \begin{figure}[!ht]
                    \centering
                    \includegraphics[width=\linewidth]{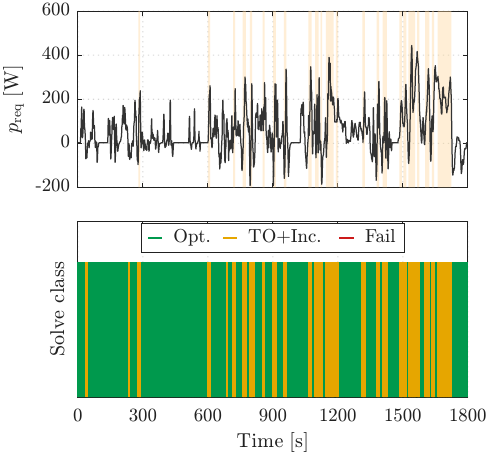}
                    \caption{Requested power and baseline-controller solve class over the nominal WLTC cycle.}
                    \label{fig:VB1_solveclass_timeline}
                \end{figure}       

                The non-certified steps are strongly clustered in time. As Fig.~\ref{fig:VB1_solveclass_timeline} shows, \(23\) of the \(30\) maximal non-certified bursts last at least five consecutive active steps, and the longest extends over \(64\) steps. This persistence is consistent with the slow evolution of neighboring MPC instances: the requested power is strongly correlated between successive samples, while the cell SOC states change only slightly over one sampling interval. Consequently, difficult operating conditions can remain difficult over multiple control updates, so a fallback strategy designed primarily for isolated deadline overruns would not eliminate sustained periods of uncertified operation.

                The additional online processing introduced by the restriction is negligible relative to the solver-time reduction. Score evaluation, sorting, and parameter construction require \(0.6~\mathrm{ms}\) at the \(95^{\mathrm{th}}\) percentile. Measured end to end, the complete proposed-controller update occupies approximately \(40\%\) of the sampling period at the \(95^{\mathrm{th}}\) percentile and \(51\%\) in the worst case. The computational advantage therefore remains substantial after accounting for ranking and problem construction, while retaining timing margin within the \(1~\mathrm{s}\) control interval.

            \subsubsection{Same-State Approximation Cost}
            \label{sec:per_step_quality}
            
                To quantify the finite-horizon approximation cost introduced by excluding non-prefix subsets, we solve the ranked-prefix and full subset-selection formulations from identical pre-control states. Table~\ref{tab:same_state_summary} reports the resulting differences on the proposed-controller and rule-based-driver trajectories.

                \begin{table}[!ht]
                    \centering
                    \footnotesize
                    \caption{Same-State Approximation Cost on Jointly Certified, Uncurtailed Steps \(\mathcal K_{\mathrm{pair}}\).}
                    \label{tab:same_state_summary}
                    \setlength{\tabcolsep}{3pt}
                    \begin{tabular}{lcc}
                    \hline\hline
                    Metric
                    & \shortstack{Proposed-controller\\trajectory}
                    & \shortstack{Rule-based-driver\\trajectory} \\
                    \hline
                    \(|\mathcal K_{\mathrm{pair}}|\) (of 1800)
                    & 1649 & 1650 \\

                    \(\Delta\bar J_{\mathrm{obj}}^{\mathrm{SS}}\)
                    & 0.0059 & 0.0032 \\

                    \(\Delta\bar J_{\mathrm{SOH},w}^{\mathrm{SS}}\)
                    & \(+0.0068\) & \(+0.0034\) \\

                    \(\Delta\bar J_{\mathrm{SOC},w}^{\mathrm{SS}}\)
                    & \(-0.0009\) & \(-0.0002\) \\

                    \(\bar g_{\mathrm{rel}}\) [\%]
                    & 0.102 & 0.056 \\

                    \(g_{\mathrm{rel},50}\) [\%]
                    & 0.000 & 0.000 \\

                    \(g_{\mathrm{rel},95}\) [\%]
                    & 0.629 & 0.259 \\

                    \(g_{\mathrm{rel},\max}\) [\%]
                    & 0.948 & 0.458 \\
                    \hline\hline
                    \end{tabular}
                \end{table}

                The median restriction gap is zero to the reported precision on both evaluated trajectories: on at least half of the paired states the two formulations attain optimal values differing by less than the \(10^{-4}\) relative MIP-gap tolerance to which both are solved. This is consistent with Proposition~\ref{prop:exactness}: at those states the ranked-prefix feasible set attains the full subset-selection optimum to solver accuracy. The same result on the independently generated rule-based-driver trajectory shows that this behavior is not confined to states induced by the proposed controller.

                Where the restriction changes the optimum, the objective loss remains small. The higher SOH-weighted contribution is partly offset by a lower SOC-weighted contribution. Neither formulation curtails power on the paired states, so \(\Delta\bar J_{\lambda,w}^{\mathrm{SS}}=0\). Thus, the restriction produces a small redistribution of the SOH--SOC tradeoff without affecting power delivery in these comparisons.

                Nevertheless, the baseline controller remains uncertified at \(151\) of the \(1800\) active states on the proposed-controller trajectory and \(150\) of the \(1800\) states on the rule-based-driver trajectory, even with the \(100~\mathrm{s}\) reference limit. The restriction-gap statistics therefore apply only to the jointly certified states in \(\mathcal K_{\mathrm{pair}}\).

            \subsubsection{Closed-Loop Control Performance}
            \label{sec:closed_loop}

                We next examine how the proposed controller performs when deployed in closed loop under the operational solver limit. Table~\ref{tab:closed_loop_summary} compares the resulting closed-loop performance of the baseline and proposed controllers.
                
                \begin{table}[!ht]
                    \centering
                    \footnotesize
                    \caption{Closed-Loop Performance Over The Nominal WLTC Cycle.}
                    \label{tab:closed_loop_summary}
                    \setlength{\tabcolsep}{4pt}
                    \begin{tabular}{lcc}
                    \hline\hline
                    Metric & Baseline & Proposed \\
                    \hline
                    \multicolumn{3}{l}{\textit{Closed-loop online objective}}\\
                    \(\bar J_{\mathrm{obj}}^{\mathrm{CL}}\) [--]
                    & 5.8521 & 5.7704 \\
    
                    \(\bar J_{\mathrm{SOH},w}^{\mathrm{CL}}\) [--]
                    & 5.7656 & 5.7627 \\
    
                    \(\bar J_{\mathrm{SOC},w}^{\mathrm{CL}}\) [--]
                    & 0.0837 & 0.0077 \\
    
                    \(\bar J_{\lambda,w}^{\mathrm{CL}}\) [--]
                    & 0.0029 & 0.0000 \\
    
                    \hline
                    \multicolumn{3}{l}{\textit{Curtailment}}\\
                    \(\Phi_{\mathrm{curt}}\) [\%]
                    & 0.278 & 0.000 \\
    
                    \hline
                    \multicolumn{3}{l}{\textit{Cell-to-cell SOC regulation}}\\
                    \(\bar{\sigma}_{\mathrm{SOC}}\) [\%]
                    & 0.8099 & 0.2403 \\
    
                    \(\sigma_{\mathrm{SOC},95}\) [\%]
                    & 1.2578 & 0.4322 \\
    
                    \(\sigma_{\mathrm{SOC},\max}\) [\%]
                    & 1.3702 & 0.5150 \\
    
                    \hline
                    \multicolumn{3}{l}{\textit{Power-tracking error}}\\
                    \(e_{50}\) [\%]
                    & 0.0154 & 0.0157 \\
    
                    \(e_{99}\) [\%]
                    & 0.0291 & 0.0288 \\
    
                    \(e_{\max}\) [\%]
                    & 1.0043 & 0.0421 \\
    
                    \hline
                    \multicolumn{3}{l}{\textit{Cell usage and switching}}\\
                    \(\bar{\phi}_{\mathrm{eng}}\) [--]
                    & 0.708 & 0.707 \\
    
                    \(\bar n_{\mathrm{sw}}\) [--]
                    & 1.81 & 2.68 \\
    
                    \(n_{\mathrm{sw},95}\) [--]
                    & 10 & 11 \\
                    \hline\hline
                    \end{tabular}
                \end{table}

                The dominant improvement is in cell-to-cell SOC regulation. The cycle-mean SOC dispersion decreases by a factor of \(3.4\), while the \(95^{\mathrm{th}}\)-percentile and maximum dispersions decrease by factors of \(2.9\) and \(2.7\), respectively. The SOH-weighted objective contribution changes by less than \(0.1\%\), indicating that the improved balancing is obtained without a material change in aggregate SOH-related cost.

                Because the two controllers engage essentially the same fraction of cells on average, the balancing improvement arises from how engagement is distributed among cells as their relative ranking evolves. Figure~\ref{fig:VB4_engagement_heatmap} illustrates this redistribution over the drive cycle.
                
                \begin{figure}[!t]
                    \centering
                    \includegraphics[width=\linewidth]
                    {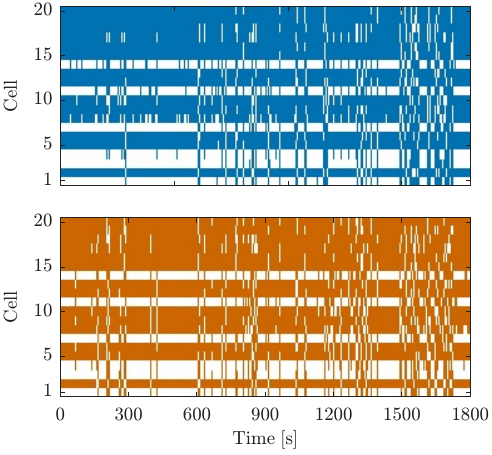}
                    \caption{Cell engagement over the nominal WLTC cycle for the baseline (top) and proposed (bottom) controllers.}
                    \label{fig:VB4_engagement_heatmap}
                \end{figure}

                Up to the \(99^{\mathrm{th}}\) percentile, the controllers have nearly identical relative tracking errors of approximately \(0.03\%\). The difference is confined to the upper tail: the baseline exhibits five detected curtailment steps and a maximum tracking error of \(1.0043\%\), whereas the proposed controller exhibits no detected curtailment and a maximum error of \(0.0421\%\). Thus, the tracking improvement is primarily the suppression of rare large deviations.
                
                The improved SOC redistribution is accompanied by increased switching activity. The mean number of changed cell states rises by \(47.9\%\), from \(1.81\) to \(2.68\), whereas the \(95^{\mathrm{th}}\) percentile increases only from \(10\) to \(11\). The increase therefore reflects more frequent routine reconfiguration, with little change in the magnitude of the largest switching events.

        \subsection{Robustness and Sensitivity Studies}
        \label{sec:robustness}
            
            We next examine whether the nominal findings persist under variations in initial conditions, cell heterogeneity, operating demand, pack size, and prediction horizon.
            
            \subsubsection{Sensitivity to Initial Conditions}
            \label{sec:VC1}
                
                We first assess sensitivity to the initial conditions. Table~\ref{tab:mc_bootstrap} reports the paired Monte Carlo results over \(20\) independently sampled initial-condition realizations.
                
                \begin{table}[!ht]
                    \centering
                    \footnotesize
                    \caption{Paired Monte Carlo Results Over 20 Initial-Condition Realizations.}
                    \label{tab:mc_bootstrap}
                    \setlength{\tabcolsep}{3pt}
                    \begin{tabular}{lcccc}
                    \hline\hline
                    Metric
                    & Baseline
                    & Proposed
                    & \shortstack{Mean paired\\difference}
                    & \shortstack{95\% bootstrap\\CI} \\
                    \hline
                    $\bar{e}$ [\%]
                    & 0.404
                    & 0.015
                    & $-0.343$
                    & $[-0.393\,,-0.293]$ \\                
                    $\bar{\sigma}_{\mathrm{SOC}}$ [\%]
                    & 0.769
                    & 0.213
                    & $-0.546$
                    & $[-0.565,\,-0.525]$ \\
                
                    $t_{95}$ [ms]
                    & 1002.0
                    & 117.7
                    & $-883.7$
                    & $[-885.3,\,-881.6]$ \\
                
                    $\eta_{\mathrm{opt}}$ [\%]
                    & 76.2
                    & 100.0
                    & $+24.0$
                    & $[+23.5,\,+24.7]$ \\
                    \hline\hline
                \end{tabular}
                
                \end{table}
            
                Across the sampled initial conditions, the proposed controller maintains lower tracking error and SOC dispersion, a substantially lower runtime tail, and a \(100\%\) certified-optimality rate. The \(95\%\) bootstrap confidence intervals for all reported paired differences exclude zero, indicating that the differences are systematic across the sampled realizations and are not attributable to a small number of cases.

            \subsubsection{Sensitivity to Cell Heterogeneity and Demand Conditions}
            \label{sec:VC2}
                
                We next vary cell heterogeneity and operating demand while holding \(N=20\), \(N_p=10\), and the remaining controller and solver settings fixed. Table~\ref{tab:robustness_sweeps} summarizes these cases together with the pack-size and prediction-horizon sweeps.

                \begin{table}[!ht]
                    \centering
                    \scriptsize
                    \caption{Robustness Results For Heterogeneity, Demand, Pack-Size, and Prediction-Horizon Variations.}
                    \label{tab:robustness_sweeps}
                    \setlength{\tabcolsep}{2.5pt}
                    \begin{tabular}{@{}lllccccc@{}}
                    \hline\hline
                    Sweep & Cond. & Ctrl. & $\bar e$ [\%] & \(\Phi_{\mathrm{curt}}\) [\%]
                        & $\bar\sigma_{\mathrm{SOC}}$ [\%] & $t_{95}$ [ms] & $\eta_{\mathrm{opt}}$ [\%]\\
                    \hline
                    \multirow{8}{*}{\shortstack[l]{Heterog.\\/demand}}
                      & \multirow{2}{*}{WLTC Low}  & Base & 0.186 & 0.167 & 0.469 & 1002.0 & 77.1 \\
                      &                            & Prop & 0.018 & 0.000 & 0.127 & 119.0  & 100.0 \\
                    \cline{2-8}
                      & \multirow{2}{*}{WLTC Med}  & Base & 0.186 & 0.222 & 0.823 & 1002.0 & 76.2 \\
                      &                            & Prop & 0.018 & 0.000 & 0.240 & 117.0  & 100.0 \\
                    \cline{2-8}
                      & \multirow{2}{*}{WLTC High} & Base & 1.630 & 0.222 & 1.252 & 1002.0 & 73.9 \\
                      &                            & Prop & 0.018 & 0.000 & 0.548 & 114.0  & 100.0 \\
                    \cline{2-8}
                      & \multirow{2}{*}{High Power$^\dagger$} & Base & 0.348 & 0.000 & 0.349 & 1007.0 & 26.8 \\
                      &                            & Prop & 0.013 & 0.000 & 0.097 & 135.5  & 100.0 \\
                    \hline
                    \hline                    
                    \multirow{6}{*}{\shortstack[l]{Pack\\size}}
                      & \multirow{2}{*}{$N=16$}    & Base & 0.023 & 0.056 & 0.735 & 1001.0 & 60.1 \\
                      &                            & Prop & 0.023 & 0.000 & 0.224 & 96.0   & 100.0 \\
                    \cline{2-8}
                      & \multirow{2}{*}{$N=20$}    & Base & 0.186 & 0.222 & 0.823 & 1002.0 & 76.2 \\
                      &                            & Prop & 0.018 & 0.000 & 0.240 & 117.0  & 100.0 \\
                    \cline{2-8}
                      & \multirow{2}{*}{$N=24$}    & Base & 0.021 & 0.000 & 0.858 & 1008.0 & 75.3 \\
                      &                            & Prop & 0.022 & 0.000 & 0.221 & 137.0  & 100.0 \\
                    \hline
                    \hline                    
                    \multirow{8}{*}{\shortstack[l]{Horizon}}
                      & \multirow{2}{*}{$N_p=4$}   & Base & 0.017 & 0.000 & 1.388 & 1000.0 & 85.9 \\
                      &                            & Prop & 0.018 & 0.000 & 0.240 & 40.0   & 100.0 \\
                    \cline{2-8}
                      & \multirow{2}{*}{$N_p=6$}   & Base & 0.129 & 0.111 & 1.171 & 1001.0 & 84.3 \\
                      &                            & Prop & 0.018 & 0.000 & 0.240 & 61.0   & 100.0 \\
                    \cline{2-8}
                      & \multirow{2}{*}{$N_p=8$}   & Base & 0.129 & 0.222 & 0.983 & 1001.0 & 81.7 \\
                      &                            & Prop & 0.018 & 0.000 & 0.240 & 93.0   & 100.0 \\
                    \cline{2-8}
                      & \multirow{2}{*}{$N_p=10$}  & Base & 0.186 & 0.222 & 0.823 & 1002.0 & 76.2 \\
                      &                            & Prop & 0.018 & 0.000 & 0.240 & 117.0  & 100.0 \\
                    \hline\hline
                    \end{tabular}
                    \vspace{2pt}
                    \parbox{\linewidth}{%
                        \footnotesize
                        ``Base'' and ``Prop'' denote the baseline and proposed controllers, respectively.\\ $^\dagger$High Power runs $K_{\mathrm{act}} = 600$ active steps, versus 1800 for all WLTC cases.
                    }
                \end{table}     
                
                Increasing the initial cell dispersion increases the residual SOC imbalance for both controllers, confirming that balancing becomes physically more difficult as cell heterogeneity increases. The proposed controller nevertheless maintains a similar runtime tail and certifies every active step across the tested WLTC heterogeneity levels.

                The High Power case produces a different stress condition: the baseline certifies only \(26.8\%\) of the active steps despite the constant demand. Its computational difficulty therefore cannot be attributed solely to temporal variation in the WLTC profile. Under the same condition, the proposed controller remains fully certified and reduces the cycle-mean SOC dispersion from \(0.349\%\) to \(0.097\%\).

            \subsubsection{Sensitivity to Pack Size}
            \label{sec:VC3}
                
                Increasing the pack size directly increases the number of cell-level engagement decisions. As shown in Table~\ref{tab:robustness_sweeps}, the proposed controller's \(t_{95}\) increases from \(96~\mathrm{ms}\) at \(N=16\) to \(137~\mathrm{ms}\) at \(N=24\), while every active step remains certified. The baseline \(t_{95}\), by contrast, remains near the \(1~\mathrm{s}\) operational limit for all tested pack sizes. Thus, the computational advantage of the proposed controller persists as the number of engagement decisions increases over the tested range.

            \subsubsection{Sensitivity to Prediction Horizon}
            \label{sec:VC4}

                Increasing the prediction horizon exposes a computation--performance tradeoff for the baseline controller. From \(N_p=4\) to \(N_p=10\), its cycle-mean SOC dispersion decreases from \(1.388\%\) to \(0.823\%\), but its certified-optimality rate decreases from \(85.9\%\) to \(76.2\%\), while curtailment appears for \(N_p\ge6\). Even at \(N_p=4\), the baseline \(t_{95}\) reaches the operational solver limit.

                For the proposed controller, every active step is certified and no detected curtailment occurs at any tested horizon. Its \(t_{95}\) increases from \(40~\mathrm{ms}\) at \(N_p=4\) to \(117~\mathrm{ms}\) at \(N_p=10\), while the cycle-mean SOC dispersion remains \(0.240\%\) to the displayed precision. Under the tested operating conditions, extending the proposed controller's horizon beyond \(N_p=4\) therefore provides no improvement in closed-loop SOC balancing.

                At \(N_p=4\), the proposed controller achieves lower cycle-mean SOC dispersion than the baseline controller at \(N_p=10\), while certifying every active step and exhibiting a \(95^{\mathrm{th}}\)-percentile solver time about \(25\) times lower. 

    \section{Conclusion}
    \label{sec:conclusion}
        This paper proposes RFSR-MI-MPC, a ranking-based feasible-set restriction for real-time MI-MPC, and instantiates it for cell-bypass reconfigurable battery packs. The resulting certificates apply to the ranked-prefix formulation, not to excluded non-prefix subsets. Moreover, the formulation remains compatible with standard mixed-integer convex solvers and requires no offline training, custom branching, or solver modification.

        In the nominal 20-cell WLTC case, the restriction increases the certified-optimality rate from \(77.61\%\) to \(100\%\), reduces the 95th-percentile solver time by a factor of \(6.8\), and decreases the cycle-mean cell-to-cell SOC standard deviation by a factor of \(3.4\) without power curtailment. Same-state comparisons yield mean restriction gaps below \(0.11\%\) on jointly certified steps, while the tested initial-condition, heterogeneity, demand, pack-size, and horizon variations retained the computational benefit without retuning. Future work should address hardware validation, switching penalties or dwell-time constraints, robustness to SOC and SOH estimation errors, and the extension of RFSR-MI-MPC to other structured MI-MPC problems.

    \appendix
        
        \section{Detailed Mixed-Integer Convex Reformulation}
        \label{app:micp}
    
            This appendix gives the algebraic mixed-integer convex reformulations referenced in Section~\ref{subsec:micp}. Unless stated otherwise, all constraints hold for \(i\in\mathcal C\) and \(\ell\in\mathcal H_k\).
    
            \subsection{Binary--Continuous Products}
    
                The step-dependent current bounds used in the product reformulations are
                \begin{subequations}
                \label{eq:packCurrentBounds}
                \begin{align}
                \bar{i}_{\mathrm{pack}}(\ell)
                &=
                \begin{cases}
                \min\left(
                \dfrac{p_{\mathrm{req}}(\ell)}{v_{\mathrm{pack}}^{\min}},\,
                i_{\mathrm{pack}}^{\max}
                \right)+\varepsilon,
                & p_{\mathrm{req}}(\ell)>0,\\[2ex]
                0,
                & p_{\mathrm{req}}(\ell)\le 0,
                \end{cases}\\[1ex]
                \underline{i}_{\mathrm{pack}}(\ell)
                &=
                \begin{cases}
                0,
                & p_{\mathrm{req}}(\ell)\ge 0,\\[2ex]
                \max\left(
                \dfrac{p_{\mathrm{req}}(\ell)}{v_{\mathrm{pack}}^{\min}},\,
                i_{\mathrm{pack}}^{\min}
                \right)-\varepsilon,
                & p_{\mathrm{req}}(\ell)<0,
                \end{cases}
                \end{align}
                \end{subequations}
                where \(\varepsilon>0\) is a small numerical safety margin.
    
                Each binary--continuous product \(w=b\xi\) with \(b\in\{0,1\}\) and \(\xi\in[\underline{\xi},\bar{\xi}]\) is represented exactly by its standard convex-hull linearization~\cite{bemporad1999control,williams2013model}. For the cell-current product, the constraints are
                \begin{subequations}
                \label{eq:currentLinearization}
                \begin{align}
                i_i(\ell)
                &\ge \underline{i}_{\mathrm{pack}}(\ell)S_i(\ell),\\
                i_i(\ell)
                &\le \bar{i}_{\mathrm{pack}}(\ell)S_i(\ell),\\
                i_i(\ell)
                &\ge i_{\mathrm{pack}}(\ell)
                -\bar{i}_{\mathrm{pack}}(\ell)\bigl(1-S_i(\ell)\bigr),\\
                i_i(\ell)
                &\le i_{\mathrm{pack}}(\ell)
                -\underline{i}_{\mathrm{pack}}(\ell)\bigl(1-S_i(\ell)\bigr).
                \end{align}
                \end{subequations}
                For the engaged-voltage product, the constraints are
                \begin{subequations}
                \label{eq:voltageMcCormick}
                \begin{align}
                v_{\mathrm{eng},i}(\ell)
                &\ge v_{\mathrm{cell}}^{\min}S_i(\ell),\\
                v_{\mathrm{eng},i}(\ell)
                &\le v_{\mathrm{cell}}^{\max}S_i(\ell),\\
                v_{\mathrm{eng},i}(\ell)
                &\ge v_i(\ell)
                -v_{\mathrm{cell}}^{\max}\bigl(1-S_i(\ell)\bigr),\\
                v_{\mathrm{eng},i}(\ell)
                &\le v_i(\ell)
                -v_{\mathrm{cell}}^{\min}\bigl(1-S_i(\ell)\bigr).
                \end{align}
                \end{subequations}
                These constraints are exact representations of \eqref{eq:cellCurrent} and \eqref{eq:engagedCellVoltage} for the stated bounds.
    
            \subsection{Quadratic Current Terms}
    
                The squared pack current is bounded below by the epigraph variable \(q_{\mathrm{pack}}\):
                \begin{subequations}
                    \label{eq:currentSquareEpi}
                    \begin{align}
                    q_{\mathrm{pack}}(\ell)
                    &\ge i_{\mathrm{pack}}^2(\ell),
                    \label{eq:currentSquare}\\
                    0 &\le q_{\mathrm{pack}}(\ell) \le \bar q_{\mathrm{pack}}(\ell),\\
                    \bar q_{\mathrm{pack}}(\ell)
                    &=
                    \begin{cases}
                    \bar{i}_{\mathrm{pack}}^2(\ell), & p_{\mathrm{req}}(\ell)>0,\\
                    \underline{i}_{\mathrm{pack}}^2(\ell), & p_{\mathrm{req}}(\ell)<0,\\
                    0, & p_{\mathrm{req}}(\ell)=0.
                    \end{cases}
                    \end{align}
                \end{subequations}
                Constraint~\eqref{eq:currentSquare} is second-order-cone representable. Since \(S_i^2=S_i\), the ohmic loss \(S_i i_{\mathrm{pack}}^2\) is represented by \(q_{\mathrm{I},i} = S_i q_{\mathrm{pack}}\). Because \(q_{\mathrm{pack}}\in[0,\bar q_{\mathrm{pack}}]\), its exact convex-hull formulation is
                \begin{subequations}
                \label{eq:heatMcCormick}
                \begin{align}
                0
                &\le q_{\mathrm{I},i}(\ell)
                \le \bar q_{\mathrm{pack}}(\ell)S_i(\ell),\\
                q_{\mathrm{I},i}(\ell)
                &\le q_{\mathrm{pack}}(\ell),\\
                q_{\mathrm{I},i}(\ell)
                &\ge q_{\mathrm{pack}}(\ell)
                -\bar q_{\mathrm{pack}}(\ell)\bigl(1-S_i(\ell)\bigr).
                \end{align}
                \end{subequations}

                The heat-generation term used in the optimizer is
                \begin{equation}
                \dot Q_{\mathrm{cvx},i}^{\mathrm{gen}}(\ell)
                =
                R_0\,q_{\mathrm{I},i}(\ell)
                +
                r_s\,q_{\mathrm{pack}}(\ell),
                \label{eq:heatGenConvex}
                \end{equation}
                which replaces \(\dot Q_i^{\mathrm{gen}}(\ell)\) in the thermal state update~\eqref{eq:cellDynamics_Thermal}. The polarization loss \(R_1 i_{\mathrm{RC},i}^{2}\) of~\eqref{eq:qi_gen} is neglected in the prediction model; the plant simulation uses the full physical expression~\eqref{eq:qi_gen}.           
    
            \subsection{SOS2 Approximation of Inverse Pack Voltage}
        
                The curtailment definition~\eqref{eq:lambdaDef} contains
                \(1/v_{\mathrm{pack}}(\ell)\). Let \(\{(\nu_j,r_j)\}_{j=1}^{N_v}\) be
                breakpoints with
                \begin{equation}
                \nu_j\in[v_{\mathrm{pack}}^{\min},v_{\mathrm{pack}}^{\max}],
                \qquad
                r_j=\frac{1}{\nu_j},
                \qquad j=1,\ldots,N_v.
                \end{equation}
                The breakpoints are placed uniformly in \(1/v\): the reciprocals
                \(r_j=1/\nu_j\) are equally spaced on
                \([1/v_{\mathrm{pack}}^{\max},\,1/v_{\mathrm{pack}}^{\min}]\),
                giving spacings in \(v\) that are non-uniform, ranging from
                \(1.34\,\mathrm{V}\) near \(v_{\mathrm{pack}}^{\min}\) to
                \(3.21\,\mathrm{V}\) near \(v_{\mathrm{pack}}^{\max}\).
                The number of breakpoints \(N_v\) is selected automatically as the
                smallest integer satisfying
                \begin{equation}
                \max_{v\in[v_{\mathrm{pack}}^{\min},v_{\mathrm{pack}}^{\max}]}
                v^2\,\bigl|\delta r(v)\bigr|
                \;\le\;
                N\,\varepsilon_{\mathrm{inv}},
                \label{eq:NvCriterion}
                \end{equation}
                where \(\widehat r(v)\) is the piecewise-linear interpolant of \(1/v\) through the breakpoints \(\{(\nu_j,r_j)\}_{j=1}^{N_v}\), evaluated at a general voltage \(v\) (distinct from its instance \(\widehat r_{\mathrm{pack}}(\ell)\) evaluated at the decision variable \(v_{\mathrm{pack}}(\ell)\) in~\eqref{eq:sos2ConvexComb}); \(\delta r(v)=\widehat r(v)-1/v\) is the resulting piecewise-linear interpolation error and \(\varepsilon_{\mathrm{inv}}=2\,\mathrm{mV}\) is the per-cell inverse-voltage error budget.
         
                The SOS2 interpolation weights \(\theta_j \ge 0\) satisfy
                \begin{subequations}
                    \label{eq:sos2ConvexComb}
                    \begin{align}
                        v_{\mathrm{pack}}(\ell)
                        &= \sum_{j=1}^{N_v}\nu_j\,\theta_j(\ell),\\
                        \widehat r_{\mathrm{pack}}(\ell)
                        &= \sum_{j=1}^{N_v}r_j\,\theta_j(\ell),\\
                        \sum_{j=1}^{N_v}\theta_j(\ell) &= 1,
                    \end{align}
                \end{subequations}
                with an SOS2 adjacency condition on \(\{\theta_j(\ell)\}\). 
                The
                approximation
                \begin{equation}
                \frac{1}{v_{\mathrm{pack}}(\ell)}\approx\widehat r_{\mathrm{pack}}(\ell)
                \label{eq:sos2Appx}
                \end{equation}
                yields the curtailment slack used in the optimizer:
                \begin{equation}
                \lambda(\ell)=
                \begin{cases}
                \dfrac{
                p_{\mathrm{req}}(\ell)\,\widehat r_{\mathrm{pack}}(\ell)
                -i_{\mathrm{pack}}(\ell)
                }{
                \operatorname{sgn}\bigl(p_{\mathrm{req}}(\ell)\bigr)
                },
                & p_{\mathrm{req}}(\ell)\neq 0,\\[1.5ex]
                0,
                & p_{\mathrm{req}}(\ell)=0.
                \end{cases}
                \label{eq:lambdaConvex}
                \end{equation}
         
                For the nominal bounds (\(v_{\mathrm{pack}}\in[44.8,72.0]\,\mathrm{V}\), \(N=20\), \(\varepsilon_{\mathrm{inv}}=2\,\mathrm{mV}\)), criterion~\eqref{eq:NvCriterion} yields \(N_v=14\) breakpoints and a worst-case point-wise inverse-voltage error \(\varepsilon_{\mathrm{SOS2}}=7.40\times10^{-6}\,\mathrm{V}^{-1}\). The implied current-equivalent curtailment error satisfies \(|\Delta\lambda(\ell)|\le|p_{\mathrm{req}}(\ell)|\,\varepsilon_{\mathrm{SOS2}} \le 3.3\times10^{-3}\,\mathrm{A}\) at the peak request \(444\,\mathrm{W}\), below the detection threshold \(\varepsilon_\lambda=0.01\,\mathrm{A}\).

            \subsection{Terminal SOC-Spread Constraint}
    
                The terminal SOC-spread constraint~\eqref{eq:termSOCSpreadConstraint}
                is linearized by two auxiliary scalar variables:
                \begin{subequations}
                \label{eq:terminalSpreadLinear}
                \begin{align}
                \underline{\mathrm{SOC}}_{\mathrm{term}}
                \le \mathrm{SOC}_i(N_p+1)
                &\le\overline{\mathrm{SOC}}_{\mathrm{term}},
                \qquad i=1,\ldots,N,\\
                \overline{\mathrm{SOC}}_{\mathrm{term}}
                -\underline{\mathrm{SOC}}_{\mathrm{term}}
                &\le\Delta \mathrm{SOC}^{\max}(k).
                \end{align}
                \end{subequations}
                This reformulation is linear and introduces no binary variables.
    
                The local affine OCV approximation used in the optimizer is described in Section~\ref{subsec:micp}. The OCV lookup data, segment-selection and extrapolation rules, and fit-error assessment are provided in~\cite{skegroRankedPrefixArchive}, because they are calibration details rather than part of the ranked-prefix contribution.

    \bibliographystyle{IEEEtran}
    \bibliography{IEEEabrv,references}
    
\end{document}